\documentclass[journal,draftcls,onecolumn,12pt,twoside]{IEEEtranTCOM}
\usepackage{color}     
\usepackage{epsf,psfrag,amssymb,amsfonts,latexsym,cite,verbatim,enumerate,subfigure}
\usepackage{srcltx,amsmath,cases, graphicx}
\usepackage{times, multirow, array}
\usepackage[mathscr]{eucal}
\usepackage{algorithm,algorithmic}

\def\psfancypar#1#2{\begingroup\def\par{\endgraf\endgroup\lineskiplimit=0pt}
               \setbox2=\hbox{\large\sc #2}
               \newdimen\tmpht \tmpht \ht2 \advance\tmpht by \baselineskip
               \font\hhuge=Times-Bold at \tmpht
               \setbox1=\hbox{{\hhuge #1}}
               \count7=\tmpht \count8=\ht1
               \divide\count8 by 1000 \divide\count7 by \count8 
               \tmpht=.001\tmpht\multiply\tmpht by \count7 
               \font\hhuge=Times-Bold at \tmpht
               \setbox1=\hbox{{\hhuge #1}}
               \noindent
                \hangindent1.05\wd1
               \hangafter=-2 {\hskip-\hangindent
               \lower1\ht1\hbox{\raise1.0\ht2\copy1}%
                \kern-0\wd1}\copy2\lineskiplimit=-1000pt}

\def\Upsilonbf{\hbox{\boldmath$\Upsilon$\unboldmath}}
\newcommand{\Phibf}{\mbox{${\bf \Phi}$}}

\newcommand{\Psibf}{\mbox{${\bf \Psi}$}}

\newcommand{\etabf}{\mbox{${\bf\eta}$}}

\newcommand{\E}{\mbox{{\rm E}}}

\def\boxit#1{\vbox{\hrule\hbox{\vrule\kern3pt
        \vbox{\kern3pt#1\kern3pt}\kern3pt\vrule}\hrule}}

\def\reals{ { {\rm  I \kern-0.15em R }  } }
\def\complex{ {\,{{\rm C} \kern-0.50em \raise0.20ex {  |}}\, }}

\def\etabf{\hbox{\boldmath$\eta$\unboldmath}}

\def\mubf{\hbox{\boldmath$\mu$\unboldmath}}

\def\xibf{\hbox{\boldmath$\xi$\unboldmath}}

\def\Sigmabf{\hbox{$\bf \Sigma$}}
\def\Upsilonbf{\hbox{$\bf \Upsilon$}}

\def\Lambdabf{\mbox{$ \bf \Lambda $}}

\def\hbf{{\bf h}}

\def\nbf{{\bf n}}

\def\sbf{{\bf s}}

\def\vbf{{\bf v}}
\def\wbf{{\bf w}}
\def\xbf{{\bf x}}
\def\ybf{{\bf y}}

\def\xbf{{\bf x}}
\def\ybf{{\bf y}}
\def\Abf{{\bf A}}
\def\Bbf{{\bf B}}

\def\Dbf{{\bf D}}

\def\Hbf{{\bf H}}
\def\Ibf{{\bf I}}

\def\Mbf{{\bf M}}

\def\Pbf{{\bf P}}

\def\Rbf{{\bf R}}

\def\Ubf{{\bf U}}
\def\Vbf{{\bf V}}
\def\Wbf{{\bf W}}
\def\Xbf{{\bf X}}

\def\Pc{{\cal P}}

\def\be{\vskip .3cm \begin{equation}}
\def\ee{\end{equation} \vskip .4cm \noindent}
\def\defeq{{\stackrel{\Delta}{=}}}
\newcommand{\R}{\mbox{$\hat {\bf R}_{N}$}}

\newcommand{\Tr}{\mbox{Tr}}

\def\Rxx{\Rbf_{\ssstyle X\kern-.1em X}}

\let\ssstyle=\scriptscriptstyle

\def\Kout{\setbox1=\hbox{\Huge\bf K}\hbox to
1.05\wd1{\hspace{.05\wd1}
}}
\def\Sout{\setbox1=\hbox{\Huge\bf S}\hbox to 1.05\wd1{\hspace{.05\wd1}
}}

  \ifx\LabelFigloaded\MYundefined\relax
  \else
   \fi

  \def\LabelFigloaded{\relax}

  \chardef\LabelFigCatAt\the\catcode`\@
  \catcode`\@=11

 \let\LabelFigwlog@ld\wlog
 \def\wlog#1{\relax}

 \ifx\\\MYundefined@
    \let\\\relax
 \fi

  \def\ms@g{\immediate\write16}

 \def\N@wif{\csname newif\endcsname }
 \def\Temp@ {\N@wif\ifIN@}
 \ifx\INN@\MYundefined@
    \else \let\Temp@\relax
 \fi
 \Temp@

  \def\IN@{\expandafter\INN@\expandafter}
  \long\def\INN@0#1@#2@{\long\def\NI@##1#1##2##3\ENDNI@
    {\ifx\m@rker##2\IN@false\else\IN@true\fi}%
     \expandafter\NI@#2@@#1\m@rker\ENDNI@}
  \def\m@rker{\m@@rker}
 
  \newtoks\Initialtoks@  \newtoks\Terminaltoks@
  \def\SPLIT@{\expandafter\SPLITT@\expandafter}
  \def\SPLITT@0#1@#2@{\def\TTILPS@##1#1##2@{%
     \Initialtoks@{##1}\Terminaltoks@{##2}}\expandafter\TTILPS@#2@}

 \def\Shifted@@#1#2#3{\setbox0=\hbox{#3}%
   \raise -\dp0\vbox {\kern-#2%
       \hbox {\kern#1\unhbox0\kern-#1}%
           \kern#2}}

 \newcount\gridcount
 \newbox\auxGridbox@ \newbox\hGridbox@ \newbox\vGridbox@
 \newbox\Labelbox@ \newbox\auxLabelbox@
 \newbox\Coordinatebox@
 \newtoks\Labeltoks@
 \newdimen\Wdd@ \newdimen\Htt@
 \newdimen\Wddd@ \newdimen\Httt@
 
 \def\Wr@{\immediate\write16}

 \newdimen\GL@wd
 \def\GridLineWidth#1{\GL@wd=#1}

 \def\gobble#1{}
 \def\EdgeErr@{\Wr@{}%
      \Wr@{\string\Edges\space argument
      1, 10, 100 or 1000 please\string!}%
      }

 \newcount\Edgect@

 \def\Sweepup#1\endSweepup{}

 \def\SetEdges@{%
    \edef\Zr@@s{\expandafter\gobble\number\Edgect@\empty}%
        \count255=0\Zr@@s\relax
        \ifnum\count255=\z@\else\EdgeErr@\show\tailtest\fi
        \count255=1\Zr@@s\relax
        \ifnum\count255=\Edgect@\relax\else\EdgeErr@\show\leadtest\fi
    \EdgGl@b\edef\Zr@s{\expandafter\gobble\Zr@@s\empty}
    \ifnum\Edgect@>\@ne\relax\EdgGl@b\let\L@Dc\empty
        \else\EdgGl@b\edef\L@Dc{\string.}\fi
    \ifnum\Edgect@>\@ne\relax
        \EdgGl@b\edef\Edgescale@##1{\divide##1 by \Edgect@}%
        \else\EdgGl@b\edef\Edgescale@##1{}\fi
    }

 \def\Edges#1{\Edgect@=#1\relax
     \let\EdgGl@b\global \SetEdges@}

 \Edges{1}

 \def\hhrule{\hrule height \GL@wd\vskip-.\GL@wd}

 \def\hRule@{%
   \advance\gridcount -2%
   \vfil\hhrule\vfil
   \llap{\smash{\raise -2.5pt
     \hbox{\L@Dc\number\gridcount\Zr@s\kern2pt}}}%
   \hhrule
   }

\def\vvrule{\vrule width \GL@wd \kern-\GL@wd}

 \def\vRule@{\advance\gridcount 2%
   \hfil\vvrule\hfil
   \setbox\auxGridbox@=\vbox to 0pt
      {\vskip \Htt@\vskip 2pt
        \hbox to 0pt{\hss\L@Dc\number\gridcount\Zr@s\hss}\vss}%
      \wd\auxGridbox@=0pt \box\auxGridbox@
   \vvrule
   }

 \def\PlaceGrid@@{\gridcount=10 
  \setbox\hGridbox@=\hbox{%
        \hbox{%
             \hskip-.4pt\vrule
             \vbox to \Htt@{%
               \offinterlineskip\parindent=\z@\relax
               \hbox to \Wdd@{\hfil}
               \hRule@\hRule@\hRule@\hRule@
               \vfil\hhrule\vfil}%
             \vrule\hskip-.4pt}
    }%
  \gridcount=0%
  \setbox\vGridbox@=\hbox{%
      \vbox{\offinterlineskip\parindent=0pt\hsize=0pt
         \vskip-.4pt\hrule%
         \hbox to \Wdd@{%
                 \vtop to \Htt@{\vfil}%
                 \vRule@\vRule@\vRule@\vRule@
                 \hfil\vvrule\hfil}%
         \hrule\vskip-.4pt}}%
  \wd\hGridbox@=0pt\ht\hGridbox@=0pt
  \wd\vGridbox@=0pt\ht\vGridbox@=0pt
  \hbox{\box\hGridbox@\box\vGridbox@}%
  }

 \def\LabelsGlobal{\def\LabGl@b{\global}}
 \def\LabelsLocal{\def\LabGl@b{}}
 \LabelsGlobal 

 \def\SetLabels#1\endSetLabels{%
   \LabGl@b\Labeltoks@={#1()\\}%
   }

 \LabGl@b\Labeltoks@={()\\}

 \def\ShowGrid{\LabGl@b\let\PlaceGrid@\PlaceGrid@@}
 \def\HideGrid{\LabGl@b\let\PlaceGrid@\relax}
 \def\Grids{\ShowGrid\LabGl@b\let\GridSwitch@\ShowGrid}
 \def\noGrids{\HideGrid\LabGl@b\let\GridSwitch@\HideGrid}

 \noGrids

 \def\bAdjust@@{%
     \setbox\auxLabelbox@=\hbox{\raise \dp\auxLabelbox@
            \box\auxLabelbox@}}
 \def\bAdjust@{\let\vAdjust@\bAdjust@@}

 \def\eAdjust@@{\dimen0=-.5\ht\auxLabelbox@
     \advance\dimen0 by .5\dp\auxLabelbox@
     \setbox\auxLabelbox@=
            \hbox{\raise\dimen0\box\auxLabelbox@}}
 \def\eAdjust@{\let\vAdjust@\eAdjust@@}

 \def\tAdjust@@{%
     \setbox\auxLabelbox@=\hbox{\raise-\ht\auxLabelbox@
            \box\auxLabelbox@}}
 \def\tAdjust@{\let\vAdjust@\tAdjust@@}

 \let\vAdjust@\relax

 \def\lAdjust@{\let\hAdjust@\rlap}
 \def\rAdjust@{\let\hAdjust@\llap}

 \let\hAdjust@\relax\let\vAdjust@\relax

 \def\FetchLabel@#1(#2)#3\\{%
     \IN@0#2@@\ifIN@
        \setbox0=\hbox{\ignorespaces#1#3\unskip}%
        \ifdim\wd0>0pt
           \ms@g{}%
           \ms@g{ !!! Bad label(s)? !!!}%
           \message{ #1(#2)#3}%
        \fi
        \def\LabelMole@##1\endFetchLabel@{%
            \IN@0()\\@##1@%
            \ifIN@\def\Temp@{\FetchLabel@##1\endFetchLabel@}%
            \else\def\Temp@{}%
            \fi
            \Temp@
           }%
     \else
       \ignorespaces#1\unskip
       \setbox\auxLabelbox@=%
         \hbox to 0pt{\hss\ignorespaces\hAdjust@
          {\ignorespaces#3\unskip}\hss}%
       \vAdjust@
       \let\hAdjust@\relax\let\vAdjust@\relax
       \AugmentLabelBox@@{#2}%
       \ht\Labelbox@=0pt\dp\Labelbox@=0pt
       \let\LabelMole@\FetchLabel@%
     \fi\LabelMole@}

 \newtoks\XYSep@ 
 \def\SetXYSeparator#1{%
     \IN@0#1@@\ifIN@\XYSep@{*}%
     \else
     \XYSep@{#1}%
     \fi
     }

 \SetXYSeparator*

 \def\AugmentLabelBox@@#1{%
     \IN@0\the\XYSep@ @#1@\ifIN@
       \SPLIT@0\the\XYSep@ @#1@%
       \setbox\Labelbox@=\hbox to 0pt{%
         \unhbox\Labelbox@
         \Shifted@@{\the\Initialtoks@\Wddd@}%
         {\the\Terminaltoks@\Httt@}%
         {\box\auxLabelbox@}}%
     \else
         \ms@g{}%
         \ms@g{ !!! Bad insertion point. !!!}%
         \message{ (#1\ this point was rejected.)}%
     \fi
    }

 \def\FetchOption@#1[#2]#3\endFetchOption@{%
    \def\temp{#1}
    \ifx\temp\empty
       \Edgect@=#2\relax
       \let\EdgGl@b\relax
       \SetEdges@
       \Cleaner@#3%
    \fi}

 \def\Cleaner@#1[@]{\Labeltoks@{#1}}
     
 \def\PlaceLabels@@{\mathsurround=0pt
     \def\Cr@{\\}%
     \let\L\lAdjust@\let\R\rAdjust@
     \let\B\bAdjust@\let\E\eAdjust@\let\T\tAdjust@
     \expandafter\FetchOption@\the\Labeltoks@[@]\endFetchOption@
     \Wddd@=\Wdd@ \Edgescale@\Wddd@ 
     \Httt@=\Htt@ \Edgescale@\Httt@
     \expandafter\FetchLabel@\the\Labeltoks@\endFetchLabel@
     \box\Labelbox@
     }%

 \let \PlaceLabels@\PlaceLabels@@

 \def\AffixLabels#1{\setbox\Coordinatebox@=\hbox{#1}%
      \Wdd@=\wd\Coordinatebox@ \Htt@=\ht\Coordinatebox@
      \advance\Htt@ \dp\Coordinatebox@
      \hbox{\copy\Coordinatebox@\kern-\Wdd@ 
           \Shifted@@{0pt}{-\dp\Coordinatebox@}%
           {\PlaceLabels@\PlaceGrid@}%
           \kern\Wdd@}%
      \GridSwitch@ 
      \LabGl@b\Labeltoks@{()\\}%
      }
 
   \let\wlog\LabelFigwlog@ld   
   \catcode`\@=\LabelFigCatAt  

                                By

              Raymond S\'eroul <A18645@FRCCSC21.BITNET>
                                and 
              Laurent Siebenmann <lcs@topo.math.u-psud.fr>
    
              VERSIONS: July 1991, Oct 1991, Jan 1992, July 1992

INTRODUCTION

      This labelling package is intended for TeX users who
rely on non-TeX sources for for their graphics inserts.  It
provides means for adding TeX labels to such inserts with a
minimum of fuss. 

       For most labels, TeX users have in the past found it
reasonably convenient to rely on non-TeX sources. Typical
occasions when an inescapable need for TeX labels seemed to
arise are

 (a) when the graphics program lacks certain exotic or complex
mathematical symbols

 (b) when the very highest typographical quality is wanted for the
labels

 (c) when labels included with the graphics fail to print, 
 and you cannot figure out why (cf. boxedeps.doc).  The labels
 provided by labelfig.tex are 100

       Since this package first appeared, many users, who in the
past scarcely dreamed of using TeX labels, have come to use
nothing but.  So it is now appropriate to add

Intoxication Warning:  TeX labels may be addictive and expensive. 

     If you have a fast preview you may disagree, and even find
that this package provides an agreeable paste-up environment; see
extra applications at end.

     Note to publishers: It is possible and convenient to ultimately
export the TeX labels produced by labelfig.tex to become an integral
part of the EPS file. This is often desired by a publisher who typically
uses an "upmarket" graphics or page layout program, with which the
staff is skilled in perfecting figures.  See Appendix I for
a recipe.

     The authors are grateful to Patrick Ion of Math Reviews for
helpful comments and encouragement.

BASIC INSTRUCTIONS

    After reading in the macro file using

preview or proof your figure with a coordinate grid printed on
top, by typing the following:

    \ShowGrid  
    \AffixLabels{<the graphics insertion>}

Here <the graphics insertion> is what you would type to insert
the graphics object alone without the grid.  This must provide
for the space around it. For example <the graphics insertion>
might well be \BoxedEPSF{MyFigure scaled 700} using the
boxedeps.tex macro package (from same source); this provides a
TeX box containing the encapsulated PostScript insert specified by
the file MyFigure. \AffixLabels{...} provides the grid (supposing
\ShowGrid is present) and later, once you have specified labels
using the grid, it will "tack on" the labels.

     The grid is a sort of (usually elongated) checkerboard of
ten rows and ten columns and its (internal) partitions are by
default numbered  .1, ... ,.9  both horizontally (X-coordinate
running left to right) and vertically (Y-coordinate running bottom
to top).  Thus the points enclosed by the grid correspond to the
points of the unit square in the cartesian "X-Y" plane, the lower
left corner corresponding to the origin (0,0).  By extrapolation,
the full page corresponds to a larger rectangle in the plane.

     These coordinates serve to position labels as follows.
Before the \AffixLabels{...} command type label specifications:

  \SetLabels
   (<X-coordinate>*<Y-coordinate>) <first label> \\
   .
   .
   .
   (<X-coordinate>*<Y-coordinate>)  <last label> \\
  \endSetLabels

Each row specifies one label and is terminated by \\.  In each
row, the position indicator comes first; it is written as a
standard cartesian point except that the X- and Y- coordinates
are separated by * rather than a comma because TeX allows a
comma as decimal point. There are no dimension units to specify
as the unit is the grid itself.

     By default, this cartesian point specifies where the middle
of the baseline of the label will be located.  However if you precede
the point by \L [or \R] the left [or right] edge of the baseline will
be located there. Similarly you may also precede the point by \T, \E,
or \B to vertically align the top equator or bottom of the label box
at the specified point.  This gives nine standard positions of
the label with respect to the insertion point --- corresponding to
the eight principle points of the compas and the center

                     \L\T     \T      \R\T

                     \L\E     \E      \R\E

                     \L\B     \B      \R\B

But this neglects the default "baseline" level of TeX,
giving potentially three more positions

                     \L    <no tag>   \R

For text, the baseline level is often the preferred. Its relation to
the others is variable. It will often coincide with the bottom level,
as happens for "X".  But it is often distinct, as for "g", in which
case you have in all 12 distinct positions rather than 9.

     It is convenient to think of this specification of label
position as attaching the label by a thumb-tack to the coordinate
grid. There are up to twelve positions of the thumb-tack on the
label, while the position of the thumb-tack on the coordinate grid is
arbitrary.  Normally, one choses the position of the thumb-tack on
the label to be the one that is the closest to the item being
labeled.  There are good reasons for this "rule of thumb":

   (a)  It facilitates correct positioning at first try.

   (b)  If the scale of the figure must be altered after labels
have been affixed, the labels have a good chance of remaining well
positioned.

   (c)  The visible grid need not extend beyond the "bounding box"
for the figure, because the best preferred position is always
(at least almost) within the bounding box .

The second reason is particularly important. Indeed it often
happens that scale has to be altered after labelling begins, in
order to either provide space for the labels, or to adjust
proportions between the labels and the figure.  (The size of labels
is unaffected by scaling.)

     Here is an artificial but self-contained test which uses
TeX rules to make a graphics object.

TEST

    Do not skip this!

 \def\FrameIt#1{\hbox{\vrule$\vcenter {\hrule\kern3pt%
             \hbox {\kern3pt #1\kern3pt}%
               \kern3pt\hrule}$\relax\vrule}}

 \def\Caption#1#2{\FrameIt{%
       \vtop {\hsize=#1\relax \parindent=0pt
         \leftskip=0pt \rightskip=0pt plus15pt
         \parfillskip=0pt
         \lineskip=1pt\baselineskip=0pt
         #2}}}

 \def\FirstQuadrant{\hbox to 100pt{\vrule\vbox to 100pt{%
        \hbox to 100pt{\hfil}\vfil\hrule}\hss}}

  \SetLabels
    \R(.5*.2) $\zeta\,\cdot$\\
    (.9*-.10) $\xi$\\
    \R(-.03*.9) $\eta$\\
    \T(.5*.9) \Caption{70pt}{%
          \it The norm of
          $g(\xi+i\eta)$ is indicated on
          contours of this invisible surface.}\\
  \endSetLabels

  \AffixLabels{\FirstQuadrant}

  \end

  Note that the coordinates to use for labels are indicated on the
edges of the grid (when visible) corresponding to the conventional
x- and y- axes of the Cartesian plane. By default the grid is
1-by-1. However, by the command \Edges{100}, you can change this
to 100-by-100 and many users find this alternative most
convenient. Place the command \Edges{...} in your style file (or
header) since its effect is is global. Other possible edge values
are 10 and 1000.

  If you use the command \Edges{...} at all, do so with care.  For
if you accidentally delete an \Edges{...} command your labels will
abruptly be badly misplaced and may logically but mysteriously
generate "dimension too big" errors under TeX and "off page" errors
under your driver.  

  You can dictate the edgescale for an individual figure by giving
the scale in brackets immediately after \SetLabels.  Thus, to
import into an article using say \Edge{100} a figure labelled using
another edgescale, say the original 1-by-1 default, you can use
\SetLabels[1]...\endSetLabels.

GETTING IT DOWN PAT

     Complicated labeling deserves the same respect as
complicated mathematics.  Do not expect it to come out perfect the
first time!  What is needed in either case is a mechanism to
repeatedly typeset troublesome pieces.

     One mechanism is always available.  One does complicated
labelling in a separate "test" file involving just the figure being
labelled;  a texpert will know how to \dump TeX's current state as
a temporary format that restarts rapidly at each retry.  Usually,
one then pastes the completed labelled figure back into the main
TeX file, but, of course, one can also \input it as an auxiliary
file.

     If you do not have a TeXpert at handy, here is a first
approximation to an efficient setup. By deletions reduce a copy
of your article to just a few lines before and after the figure.
Now label the figure, and finally, copy and paste the labelled
figure to the original article. Then copy the next figure to label
into this testbed and repeat. The TeXpert can improve the  speed
at which TeX starts up, by compiling a format specifically for
your article; just one caution: best NOT include in the format
ephemeral details of setup like \Set<mydriver>ArtSpecials (from
boxedeps.tex because this reads  figure dimensions which you may
change during your work session.

     An improved mechanism to repeatedly typeset troublesome
pieces is now available on the Macintosh; it is called LinoTeX;
see the same ftp sources.  It could be set up on many types
of computer.

     Before using labelfig.tex to attach labels to a graphics
object inserted using boxedeps.tex or BoxedArt.tex, make it a
firm rule to carefully adjust the bounding box using the trimming
commands of these packages, and also at least tentatively scale
and position the object. Beware of changing the grid inadvertently
after the labels have been positioned.  For example, correcting
the bounding box of a PostScript graphics object can foul up the
labels by changing the coordinate grid to which the labels are
attached. This is particularly true for the trimming  commands of
boxedeps.tex and BoxedArt.tex. However, as noted already, change
of scale is much less disruptive, and modest adjustments should be
well tolerated.

     Sometimes the labels protrude so far from the bounding box
of a figure that the figure has to be repositioned.  Best do this
by ad hoc spacing, say using \hglue and \vglue; altering the
bounding box would create a vicious circle.

     Remember that you are responsible for preventing labels
from overlapping. You are responsible for all label typography
including size and style. A label is really just about anything
that can be put in a TeX box. Note that spaces at the beginning
and end of labels will normally be suppressed; if you really want
them you must protect them with TeX braces.

     This package temporarily sets the \mathsurround parameter
of TeX to zero  while the labels are being affixed. This is done
because nonzero \mathsurround space would influence the position
of left and right aligned labels; then, when a texpert or printer
modifies mathsurround, diagram labeling might be disastrously
altered. There is a small price to pay involving labels that are
formatted as caption boxes including mathematics: you  may want or
need to specify an explicit mathsurround space within the caption
box; it will not influence anything outside.

     Those hostile to the use of * as separator between
the X and Y coordinates of label insertion points, are free to
impose another using \SetXYSeparator{<the new separator>}.  
Americans may prefer "," to "*" since they never use a 
comma as a decimal point; on the other hand, * may be more visible.

APPENDIX (I)  MERGING labelfig.tex LABELS INTO AN EPSF GRAPHICS OBJECT.

     As promised in the introduction, here is a recipe useful for
publishers. It works at least on Macintosh and at least for vectorized
graphics and Adobe type1 fonts.  (There is surely a similar recipe for
PCs under MSWindows.)

 (a)  Use boxedeps.tex utility to integrate the figure given by the eps
file, "x.eps" say, with a visible frame around it.  See
\ShowDisplacementBoxes command in boxedeps.tex.  To get precise results
automatically it is important to use the \Trim... commands of
boxedeps.tex making the "DisplacementBox" neatly fit the figure.

 (b)  Use the TeX printer driver and LaserWriter (versions >= 8.1.1) to
export to an EPSF the DVI page containing the integrated, labelled
figure. You now have an EPS file  "xx.eps"  that contains too much, and at
the wrong scale, and at wrong position.

 (c)  Convert the EPSF to an Adode Illustrator format EPSF using
the shareware utility called epsConvert by Sam Weiss
1993-- (currently $25).

 (d)  In Illustrator (or a compatible program), group the labels and the
"DisplacementBox"; copy them to the clipboard and paste them into "x.ps".
This step requires that all the label fonts be "visible to the Macintosh.

 (e)  Translate and scale the pasted group consisting of the labels plus
the "DisplacementBox" so as to make the "DisplacementBox" the bounding
box of (labelless) figure represented by "x.eps".  At this point the
labels will be correctly placed on the figure "x.eps".

 (f)  Ungroup and delete the "DisplacementBox".  The result is the
desired single EPS file, "x+.eps" say, It contains the original figure
plus its labels.  

     Using grouping and ungrouping appropriately in "x+.eps", a
publisher's staff can very efficiently improve label positions etc.

APPENDIX II)  SOME EXOTIC APPLICATIONS

     The grid of labelfig.tex is analogous to a light-table in
classical page makeup with wax or latex glue.  In principle, you
can use it to compose any page from its indivisible parts.  This
even has some of the artisanal charm of classical paste-up
provided you have a fast screen preview to make the process
"interactive".

     In practice labelfig.tex is a tool for nonstandard jobs.
Here are a few going beyond the labelling already discussed.

(I)  GRAPHICS INTEGRATION.

     This is accomplished by treating the imported graphics
objects as labels.  The underlying graphics object is then
typically an empty  \vbox to <dimension>{\vfill} in a TeX
\midinsert...\endinsert construction.  A label line
might be of the form

   (.1*.1) \\

The exact form of the special command varies from driver to
driver.  However, in the case of encapsulated PostScript graphics
(EPSF norm), by relying on boxedeps.tex, one can have the
following standard syntax (independant of driver  (see
boxedeps.doc for details.
  
  (.1*.1) \BoxedEPSF{MyFigure scaled <scale in mils>}\\

This may be slow since it requires TeX to read the PostScript
file to read bounding box using many complex macros.  So you
may want to try

  (.1*.1) \EPSFSpecial{MyFigure}{<scale in mils>}\\

which is fast and driver independant, but it squashes the
bounding box, normally to its lower left corner.

     Similarly for graphics of the Macintosh PICT norm ---
using BoxedArt.tex (same sources) in place of boxedeps.tex.

     This approach to integration is to be recommended when
one is assembling a composite graphics object.

 (II)  COMMUTATIVE DIAGRAM ENHANCEMENT

     Commutative diagrams or arrays of mathematical objects
connected by arrows of various sorts are common in mathematics.
The mathematical objects require the use of TeX.  Recently TeX
acquired a good collection of arrows of all slopes --- that of
LamSTeX --- plus pwerful macros to build the diagrams.

     However, even the LamSTeX collection is often
inadequate; it lacks for example double shafted arrows, dotted
arrows and curved arrows. Fortunately it is possible to produce
such arrows on an individual basis using sophisticated graphics
programs such as Illustrator and AldusFreehand (both serving
the EPSF norm) or using Metafont (with its public domain norm).
Since the creation of each new arrow is a work of love, you
probably want to limit the number of arrows by using LamSTeX
for most arrows. The 40K commutative diagram module of LamSTeX
has been adapted to work with AmSTeX and a copy may be posted
with LabelFig and related files. Unfortunately no one has yet
offered a version that works with Plain TeX or LaTeX.

       Suffice it here to say that when the exotic arrow has
been somehow imported into TeX, labelfig.tex treats it as a
label that one affixes to the commutative diagram.  Two other
steps will be treated in separate notes, namely the matter of
extracting the dimension specifications for the arrow and the
construction of the arrow --- for these steps are far from
unique and often depend intimately on your computer environment. 
Notes for the Macintosh-Textures-Illustrator combination are
found in the file ExoticArrows.doc.

 (III) NESTING 

Ingenuity pays off in exploiting labelfig.tex. One can
mix graphics and typography quite freely.  labelfig.tex is good
for freeform or overlapping arrangements, while boxedeps.tex (or
BoxedArt.tex) is best for regimented non-overlapping
arrangements --- and the two can be combined.

     The default behavior of labelfig.tex is not ideal 
for nesting objects, because to prevent trouble for beginners
the register for labels is globally cleared when \AffixLabels
concludes.  But there are switches available

      \LabelsGlobal      \LabelsLocal

which change this.  To understand this, extend the above test 
by something like:

 \LabelsLocal

 \SetLabels
    (.5*.5) AAA\\
 \endSetLabels

 {
 \SetLabels
    (.5*.5) ZZZ\\
 \endSetLabels
   \AffixLabels{\FirstQuadrant}
 }

   \AffixLabels{\FirstQuadrant}

     There are however potential pitfalls.  Neither
labelfig.tex nor boxedeps.tex has been tested under extreme
conditions. Problems may occur if their procedures are
indiscriminately nested. For boxedeps.tex (not labelfig.tex)
there is a precise cause for worry, namely many of its
variables are "global", which means that TeX braces will not
provide the protection one might expect.

COMMAND SUMMARY FOR labelfig.tex

  Here [...] means optional (one or zero)
       [...]* means any number of such constructs

  \SetLabels
    [[<P>](<X><Sep><Y>) <label> \\]*
  \endSetLabels
  \ShowGrid  
  \AffixLabels{<the figure>}

   --- <P> is tack position, one of eleven or empty
              order irrelevant

                   \L\T      \T      \R\T

                   \L\E      \E      \R\E

                     \L               \R

                   \L\B      \B      \R\B

   --- (<X><Sep><Y>) insertion point;
  <Sep> is separator, = * by default;
  \SetXYSeparator{<Sep>} changes it.
   <X> and <Y> are real numbers

  --- <label> a label to attach 

  --- <the figure> the figure to label 

  \GlobalLabels (default)     
  \LocalLabels  setting for nested constructs.

 \Grids makes ALL grids appear; \HideGrid then makes just next disappear.
 \noGrids returns to default.  The commands are always global.

 \GridLineWidth{<dimension>} adjusts width of grid lines. Default is very
small, to give "hairline" effect. If your grid lines are missing try
setting \GridLineWidth{1pt}.

 \Edges#1 globally changes the edge size of all grids to the numerical 
value #1, which must be 1, 10, 100, or 1000.  The default is 1.

VERSION HISTORY.
 --- Jan 1993: \Edges#1 and [??] option after \SetLabels
 --- July 1992: \Grids, \noGrids, \HideGrid;
       Gridlines become hairlines; \GridLineWidth{<dimension>}.
 --- Oct 1991, Jan 1992: \SetXYSeparator{<Sep>},  \LabelsGlobal,
       \LabelsLocal.
 --- July 1991: first release

Address for bugs and other feedback:

        Raymond S\'eroul
        IREM and Lab. de Typographie Informatise
        Univ. Rene Descartes
        Strasbourg

    Tel 33-88-41-63-45
    Email:  A18645@FRCCSC21.BITNET

        Laurent Siebenmann
        Mathematique, Bat. 425,
        Univ de Paris-Sud,
        91405-Orsay,
        France

    Tel 33-1-6941-7949; 
    Email: lcs@topo.math.u-psud.fr

\def\scalefig#1{\epsfxsize #1\textwidth}
\def\defeq{\stackrel{\Delta}{=}}

\def\tcr{\textcolor{red}}

\def\wt{\widetilde}

\def\Phibf{\boldsymbol{\Phi}}
\def\xibf{\boldsymbol{\xi}}
\def\Psibf{\boldsymbol{\Psi}}
\def\Upsilonbf{\boldsymbol{\Upsilon}}
\def\etabf{\boldsymbol{\eta}}

\def\tcr{\textcolor{red}}

\def\bf{\textbf}
\def\ovl{\overline}
\def\wt{\widetilde}

\def\Lambdabf{\boldsymbol{\Lambda}}

\def\Sigmabf{\boldsymbol{\Sigma}}
\def\Tr{\text{Tr}}
\def\diag{\text{diag}}

\newcommand {\Ebb}{{\mathbb{E}}}

\newcommand{\bzero}{\bf{0}}

\newtheorem{theorem}{Theorem}

\newtheorem{lemma}{Lemma}

\normalsize

\begin{document}

\title{ \LARGE Two-Stage Beamformer Design for Massive MIMO Downlink By Trace Quotient Formulation}

\author{  Donggun Kim, Gilwon Lee, {\em Student~Members, IEEE}, and Youngchul
Sung$^\dagger$\thanks{$^\dagger$Corresponding author}, {\em
Senior~Member, IEEE}  \\
\thanks{The authors are with Dept. of Electrical Engineering,  KAIST, Daejeon 305-701, South
Korea. E-mail:\{dg.kim@, gwlee@, ysung@ee.\}kaist.ac.kr.
{This research was supported in part by Basic Science Research Program through the National Research Foundation of Korea (NRF) funded by the Ministry of Education (2013R1A1A2A10060852). }}
}

\markboth{Submitted to IEEE Transactions on Communications, \today}%
{}

\maketitle

\begin{abstract}
In this paper, the problem of outer beamformer design  based only
on channel statistic information is considered for two-stage
beamforming for multi-user massive MIMO downlink, and the problem
is approached based on signal-to-leakage-plus-noise ratio (SLNR).
To eliminate the dependence on the instantaneous channel state
information, a lower bound on the average SLNR is derived by
assuming zero-forcing (ZF) inner beamforming, and an outer
beamformer design method that maximizes the lower bound on the
average SLNR is proposed. It is shown that the proposed SLNR-based
outer beamformer design problem reduces to a trace quotient
problem (TQP), which is often encountered in the field of machine
learning. An iterative algorithm is presented to obtain an optimal
solution to the proposed TQP. The proposed method has the
capability of optimally controlling the weighting factor between
the signal power to the desired user and the interference leakage
power to undesired users according to different channel
statistics. Numerical results show that the proposed outer
beamformer design method yields significant performance gain over
existing methods.
\end{abstract}


\section{Introduction}

The multiple-input multiple-output (MIMO) technology has prevailed
in wireless communications for more than a decade. The technology
has been adopted in many wireless standards since it improves the
spectral efficiency and reliability of wireless communication
without requiring additional bandwidth. Recently, the MIMO
technology based on large-scale antenna arrays at base stations,
so-called massive MIMO, is considered to further improve the
system performance for upcoming 5G wireless systems and vigorous
research is going on on this topic. Massive MIMO can support high
data rates and energy efficiency and simplify receiver processing
based on the asymptotic orthogonality among user channels based on
large antenna arrays
\cite{Marzetta:10WCOM,Rusek&Persson&Lau&Larsson&Edfors&Tufvesson&Marzetta:13SPM}.
However, realizing the  benefits of massive MIMO in practical
systems faces several challenges especially in widely-used
frequency division duplexing (FDD) scenarios. In contrast to
current small-scale MIMO systems,  downlink channel estimation is
a difficult problem for FDD massive MIMO systems since the number
of available training symbols required for downlink channel
estimation is limited by the channel coherence time and the number
of channel parameters to estimate is very large
\cite{Noh&Zoltowski&Sung&Love:14JSTSP,
Choi&Love&Bidigare:14JSTSP,So&Kim&Lee&Sung:15SPL,Noh&Zoltowski&Love:14Arxiv}.
Furthermore, channel state information (CSI) feedback overhead for
downlink user scheduling for massive FDD multi-user MIMO can be
overwhelming without some smart structure on massive MIMO systems.
To overcome these difficulties associated with massive MIMO,  {\em
two-stage beamforming} for massive MIMO under the name of ``Joint Spatial Division and Multiplexing (JSDM)'' has been studied in
\cite{Adhikary&Nam&Ahn&Caire:13IT,
Nam&Adhikary&Ahn&Caire:14JSTSP,Lee&Sung:14ITsub,Chen&Lau:14JSAC,
Liu&Lau:14SP}.  The two-stage beamforming idea is basically a
divide-and-conquer approach, and the key ideas of the two-stage
beamforming strategy are 1) to partition the user population
supported by the serving base station into multiple groups each
with approximately the same channel covariance matrix (this can be
viewed as virtual sectorization) and 2) to decompose the MIMO
beamformer at the base station into two steps: an {\em outer
beamformer} and an {\em inner beamformer}, as shown in Fig.
\ref{fig:two_stage_bf}. The outer beamformer faces the antenna
array and roughly distinguishes different groups by bolstering
in-group transmit power and suppressing inter-group interference,
and the inner beamformer views the product of the actual channel
and the outer beamformer as an effective channel, separates the
users within a group, and provides spatial multiplexing among
in-group users \cite{Adhikary&Nam&Ahn&Caire:13IT}. Here, major
complexity reduction results from the approach that {\em the outer
beamformer is properly designed based only on channel statistic
information not on CSI}. In this case, the actually required CSI
for the inner beamformer adopting typical zero-forcing (ZF) or
regularized ZF (RZF)  beamforming is significantly reduced since
it only requires the CSI of the effective channel with
significantly reduced dimensions.

Several researchers followed the aforementioned framework for
two-stage beamforming for massive MIMO. They  adopted linear
beamforming such as ZF for the inner beamformer and tackled the
problem of outer beamformer design based on channel statistic
information \cite{Adhikary&Nam&Ahn&Caire:13IT, Chen&Lau:14JSAC,
Liu&Lau:14SP}. In \cite{Adhikary&Nam&Ahn&Caire:13IT}, Adhikary
{\it et al.} proposed a simple block diagonalization (BD)
algorithm for the outer beamformer design, which  obtains the
outer beamformer by  projecting the dominant eigenvectors of the
desired group channel covariance matrix onto the null space of the
dominant eigenspace of all other group channel covariance
matrices. In \cite{Chen&Lau:14JSAC}, Chen and Lau considered the
outer beamformer design criterion of minimizing the total
inter-group interference power minus the weighted total desired
group signal power. In this case,  for a given weighting value
between the total inter-group interference power and the total
in-group signal power, the outer beamformer is given by a set of
dominant eigenvectors of the weighted difference  between the
total undesired group channel covariance matrix sum and the
desired group channel covariance matrix. In \cite{Liu&Lau:14SP},
Liu and Lau considered the outer beamformer design from a fairness
perspective.  In this work, they designed the outer beamformer by
choosing a set of columns from a discrete Fourier transform (DFT)
matrix to maximize the minimum average rate among all the users.

In this paper, we also consider the outer beamformer design based
only on channel statistic information for the aforementioned
two-stage beamforming framework for massive MIMO. As already shown
in the previous works, computation of the
signal-to-interference-plus-noise-ratio (SINR) for each receiver
is difficult in this downlink scenario with interfering groups. To
circumvent this difficulty, as our design criterion we adopt the
average signal-to-leakage-plus-noise ratio (SLNR) criterion
\cite{Sadek&Tarighat&Sayed:07WCOM}, which is shown to be
Pareto-optimal in the achievable rate region in certain
interference channel cases \cite{Zakhour&Gesbert:09WSA,
Park&Lee&Sung&Yukawa:13SP}, and propose an average SLNR-based
outer beamformer design framework in single cell massive MIMO
systems.\footnote{The multi-cell scenario considered in
\cite{Chen&Lau:14JSAC} can be cast into this single-cell
multi-group setting simply by considering each base station in the
multi-cell case as one group in the single-cell multi-group case.}
The signal power to the desired receiver and the leakage power to
other undesired receivers by the transmitter required for the SLNR
method cannot be computed by considering only the outer
beamformer. Instead, both the outer beamformer and the inner
beamformer should jointly be considered to derive the two
quantities. Thus, to simplify analysis we assume a ZF beamformer
with equal power allocation for the inner beamformer although
simulation is performed for both ZF and regularized ZF (RZF) inner
beamformers. Even with this assumption of ZF for the inner
beamformer, the derivation of average SLNR is not straightforward
due to the joint nature. Thus, exploiting the fact that ZF is used
for the inner beamformer, we derive a lower bound on the average
SLNR that is a function of only channel statistics and the outer
beamformer, and our design criterion is to maximize this lower
bound on the average SLNR under the constraint that the outer
beamformer matrix has orthonormal columns.\footnote{The outer
beamformer matrix's having orthonormal columns is very desirable
for effective downlink channel estimation
\cite{Noh&Zoltowski&Sung&Love:14JSTSP} and downlink user
scheduling \cite{Nam&Adhikary&Ahn&Caire:14JSTSP, Lee&Sung:14ITsub}
purposes under the two-stage beamforming framework for massive
MIMO.} Then, we cast this constrained optimization problem as a
{\em trace quotient problem} (TQP), which is often encountered in
the field of pattern recognition, computer vision, and machine
learning \cite{Shen&Diepold&Huper:10MTNS,
Zhang&Yang&Liao:14OL,Wang&Yan&Xu&Tang&huang:07CVPR}. To obtain an
optimal solution to the formulated TQP, we modify the algorithm in
\cite{Wang&Yan&Xu&Tang&huang:07CVPR} to fit into the considered
case and show the optimality and convergence of the modified
algorithm based on existing results
\cite{Shen&Diepold&Huper:10MTNS,
Zhang&Yang&Liao:14OL,Wang&Yan&Xu&Tang&huang:07CVPR}. Numerical
results show that the proposed outer beamformer design approach
yields significant sum rate gain over the existing algorithms in
\cite{Adhikary&Nam&Ahn&Caire:13IT} and \cite{Chen&Lau:14JSAC}.

{\em Notation and Organization}  We will make use of standard
notational conventions. Vectors and matrices are written in
boldface with matrices in capitals. All vectors are column
vectors.  For a matrix $\Xbf$, $\Xbf^*$, $\Xbf^T$, $\Xbf^H$,
$[\Xbf]_{i,j}$, and $\mbox{Tr}(\Xbf)$ indicate the complex
conjugate, transpose, conjugate transpose, $(i,j)$-th element, and
trace of $\Xbf$, respectively. $\Ibf_n$ stands for the identity
matrix of size $n$ (the subscript is omitted
    when unnecessary). For vector $\xbf$, $||\xbf||$ represents the 2-norm of $\xbf$. $\diag(x_1, x_2, \cdots, x_n)$ means a diagonal
matrix with diagonal entries $x_1, x_2, \cdots, x_n$. The notation
$\xbf\sim {\cal{CN}}(\mubf,\Sigmabf)$ means that $\xbf$ is complex
circularly-symmetric Gaussian distributed with mean vector $\mubf$
and covariance matrix $\Sigmabf$.
$\Ebb\{\cdot\}$ denotes the
expectation. $\iota \defeq \sqrt{-1}$ and ${\mathbb{C}}$ is the
set of complex numbers.

The remainder of this paper is organized as follows. The system
model is described in Section \ref{sec:systemmodel}. In Section
\ref{sec:OuterBeamformerDesign}, a lower bound on the average SLNR
is derived and  the outer beamformer design problem is formulated
as a TQP. An iterative algorithm for the TQP is presented and its
optimality and convergence are shown. The performance of the
proposed algorithm is investigated in Section
\ref{sec:numericalresults}, followed by the conclusion in Section
\ref{sec:conclusion}.

\section{System Model}
\label{sec:systemmodel}

\begin{figure}[t]
\begin{psfrags}
        \psfrag{d1}[c]{\small $\sbf_1$} %
        \psfrag{dG}[c]{\small $\sbf_G$} %
        \psfrag{W1}[c]{\small $\Wbf_1$} %
        \psfrag{WG}[c]{\small $\Wbf_G$} %
        \psfrag{b1}[c]{\small $M_1$} %
        \psfrag{bG}[c]{\small $M_G$} %
        \psfrag{vd}[c]{\small $\vdots$} %
        \psfrag{V1}[c]{\small $\Vbf_1$} %
        \psfrag{VG}[c]{\small $\Vbf_G$} %
        \psfrag{1}[c]{\small $1$} %
        \psfrag{2}[c]{\small $2$} %
        \psfrag{M}[c]{\small $M$} %
        \psfrag{S1}[c]{\small $S_1$} %
        \psfrag{SG}[c]{\small $S_G$} %
        \psfrag{group1}[c]{\small Group $1$} %
        \psfrag{group2}[c]{\small Group $2$} %
        \psfrag{group3}[c]{\small Group $3$} %
        \psfrag{groupG}[c]{\small Group $G$} %
        \psfrag{base}[c]{\small Base station} %
    \centerline{ \scalefig{0.75} \epsfbox{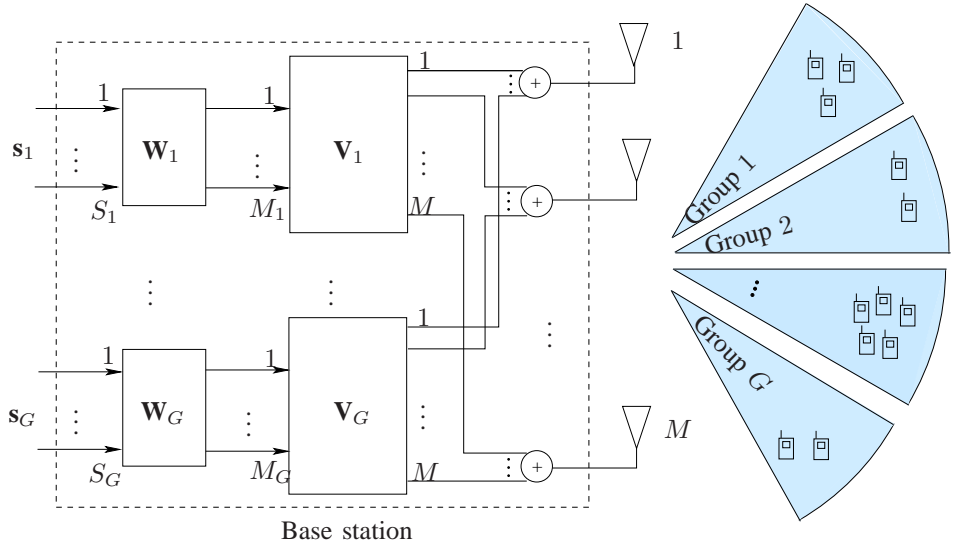} }
    \caption{Multi-group massive MIMO downlink with
    two-stage beamforming \cite{Lee&Sung:14ITsub}}
    \label{fig:two_stage_bf}
\end{psfrags}
\end{figure}

We consider a single-cell massive MIMO downlink system in which a
base station with a uniform linear array (ULA) of  $M$ transmit
antennas serves $K$ single-antenna users. We assume that $K$ users
in the cell are partitioned into $G$ groups  such that $K =
\sum_{g=1}^G K_g$, where $K_g$ is the number of users in group
$g$,  and $K_g$ users in group $g$ have the same $M \times M$
channel covariance matrix $\Rbf_g$ as in
\cite{Adhikary&Nam&Ahn&Caire:13IT, Chen&Lau:14JSAC, Liu&Lau:14SP,
Lee&Sung:14ITsub}. We assume a typical spatial correlation channel
model \cite{Kotecha&Sayeed:04SP, Choi&Love&Bidigare:14JSTSP}. That
is, the channel of user $k$ in group $g$ is given by
\begin{equation} \label{channel}
\hbf_{g_k} = \Rbf_g^{1/2} \ovl{\hbf}_{g_k},
\end{equation}
where $\ovl{\hbf}_{g_k} \in {\mathbb{C}}^{M \times 1}
\stackrel{i.i.d.}{\sim} {\cal{CN}}(\bzero, \Ibf_M).$ We assume
that the outer beamformer implements virtual sectorization, i.e.,
the outer beamformer divides the overall azimuthal angle
(typically 120 degrees in conventional cellular networks) into
multiple virtual sectors with roughly\footnote{Each sector cannot
be completely separated because the number of transmit antennas is
finite and there always exists some overlap among virtual
subsectors, even in the case we design the main coverage angle of
$2\Delta$ of each sector to be small.} 2$\Delta$ degrees for
azimuthal coverage for each virtual sector. A channel model
considering such a situation is the one-ring scattering model
\cite{Abdi&Kaveh:02JSAC, Zhang&Smith&Shafi:07WCOM}, which captures
the base station's elevation and local scattering around the
users. Under this model, the channel covariance matrix for each
sector covering $2\Delta$ azimuthal angle centered at $\theta$ can
be precomputed as \cite{Abdi&Kaveh:02JSAC,
Zhang&Smith&Shafi:07WCOM}
\begin{equation} \label{onering}
\left[ \Rbf_g \right]_{k,l} = \frac{1}{2\Delta} \int_{\theta -
\Delta}^{\theta + \Delta} e^{-\iota 2 \pi (k-l) D \sin\omega}
d\omega,
\end{equation}
where $\lambda_c$ is the carrier wavelength,  $\lambda_c D$ is the
antenna spacing, $\theta$ is the angle of the center of the
subsector, and $\Delta$ is the angle spread (AS). From here on, we
will assume that the channel covariance matrices
$\Rbf_g,g=1,\cdots,G$ are given to the base station.

As in \cite{Adhikary&Nam&Ahn&Caire:13IT,
Nam&Adhikary&Ahn&Caire:14JSTSP,Lee&Sung:14ITsub,Chen&Lau:14JSAC,
Liu&Lau:14SP},  we consider  two-stage beamforming  for downlink
transmission with outer beamformers $\{\Vbf_g,~g=1,\cdots,G\}$ for
group separation or virtual sectorization and an inner beamformer
$\Wbf_g$ for user separation within group $g$ for each
$g=1,\cdots,G$, as shown in Fig. \ref{fig:two_stage_bf}. Denote
the overall $K \times M$ channel matrix as $\Hbf = \left [
\Hbf_1^H, \Hbf_2^H, \cdots, \Hbf_G^H\right ]^H$, where $\Hbf_g =
[\hbf_{g_1}, \hbf_{g_2}, \cdots, \hbf_{g_{K_g}} ]^H$ is the $K_g
\times M$ channel matrix for the users in group $g$. Then, the
signal vector received by all the users in the cell is given by
\begin{equation}  \label{eq:firstdatamodel}
\ybf = \Hbf \Vbf \Wbf \sbf + \nbf,
\end{equation}
 where
 the overall outer beamformer matrix $\Vbf$ is partitioned into $G$
submatrices as $\Vbf = [ \Vbf_1, \Vbf_2,$ $\cdots,\Vbf_G]$ with
$\Vbf_g \in {\mathbb{C}}^{M \times M_g}$ satisfying\footnote{The
orthogonality constraint is desirable for random beamforming type
user scheduling \cite{Nam&Adhikary&Ahn&Caire:14JSTSP} or ReDOS-PBR
user scheduling \cite{Lee&Sung:14ITsub} for two-stage beamforming
based massive MIMO.}
\begin{equation}
\Vbf_g^H \Vbf_g = \Ibf;
\end{equation}
 the overall inner
beamformer $\Wbf$ has a block diagonal structure as $\Wbf = \diag(
\Wbf_1, \Wbf_2, \cdots, \Wbf_G)$ with the inner beamformer $\Wbf_g
= [ \wbf_{g_1}, \wbf_{g_2}, \cdots, \wbf_{g_{K_g}}] \in
{\mathbb{C}}^{M_g \times S_g}$ for group $g$; $\sbf = [\sbf_1^H,
\sbf_2^H, \cdots,$  $\sbf_g^H]^H$ $\sim {\cal{CN}}(\bzero,
\Ibf_S)$ is the data vector with $\sbf_g \sim {\cal{CN}}(\bzero,
\Ibf_{S_g})$; and $\nbf = [\nbf_1^H, \nbf_2^H, \cdots, \nbf_G^H]^H
\sim {\cal{CN}}(\bzero, \sigma^2\Ibf_K)$ is the noise vector.
Thus, $M_g$ is the dimension of the effective MIMO channel seen by
the inner beamformer $\Wbf_g$, and $S_g$ is the number of data
streams for group $g$. We assume that the base station has an
average transmit power constraint $\Tr(\Vbf \Wbf \Wbf^H \Vbf^H)
\le P_T$.

Combining $\Hbf\Vbf$, we can rewrite the data model
(\ref{eq:firstdatamodel})  as
\begin{equation}
\ybf =\wt{\Hbf} \Wbf \sbf + \nbf,
\end{equation}
where
\begin{equation}
\wt{\Hbf} \triangleq \Hbf\Vbf = \left[
\begin{array}{cccc}
\Hbf_1 \Vbf_1 & \Hbf_1 \Vbf_2 & \cdots & \Hbf_1 \Vbf_G \\
\Hbf_2 \Vbf_1 & \Hbf_2 \Vbf_2 & \cdots & \Hbf_2 \Vbf_G \\
\vdots        & \vdots        & \ddots &  \vdots       \\
\Hbf_G \Vbf_1 & \Hbf_G \Vbf_2 & \cdots & \Hbf_G \Vbf_G \\
\end{array}
\right]
\end{equation}
and
 $\wt{\Hbf}_g
\triangleq \Hbf_g \Vbf_g \in {\mathbb{C}}^{K_g \times M_g}$ is the
effective MIMO channel seen by the inner beamformer $\Wbf_g$ for
group $g$. We assume that the CSI of the effective MIMO channel
$\wt{\Hbf}_g$, $g=1,\cdots,G$, is available to the transmitter
(please see \cite{Noh&Zoltowski&Sung&Love:14JSTSP}) and the inner
beamformer $\Wbf_g \in {\mathbb{C}}^{M_g \times S_g}$ ($M_g \ge
S_g$) for each $g = 1, 2, \cdots, G$ is designed as a zero-forcing
beamformer with equal power $||\wbf_{g_k}||^2 = 1$ for each user
based on the effective CSI $\wt{\Hbf}_g$, i.e.,
\begin{equation}
\Wbf_{g} =  \Vbf_g^H \Hbf_{g}^H \left( \Hbf_{g} \Vbf_{g} \Vbf_g^H
\Hbf_{g}^H \right)^{-1} \Pbf_g
\end{equation}
where $\Pbf_g = \diag( \sqrt{P_{g_1}}, \cdots$, {\small ${\sqrt{P_{ g_{S_g}}}}$}), and  $P_{g_k}$ is the transmit power scaling factor for the user $k$ in the $g$-th group satisfying $||\wbf_{g_k}||^2 = 1$.
The received signal vector for the users in group $g$ is given by
\begin{equation}
\ybf_g = \Hbf_g \Vbf_g \Wbf_g \sbf_g + \sum_{g' \neq g} \Hbf_g
\Vbf_{g'} \Wbf_{g'} \sbf_{g'} + \nbf_g,
\end{equation}
where $\sbf_g = [s_{g_1}, s_{g_2}, \cdots s_{g_{S_g}}]^T \in
{\mathbb{C}}^{S_g \times 1}$ and $\nbf_g  = [n_{g_1}, n_{g_2},
\cdots n_{g_{K_g}}]^T \in {\mathbb{C}}^{K_g \times 1}$ are the
data and noise vectors for group $g$, respectively.  The received
signal of user $k$ in group $g$ is given by
\begin{equation} \label{signalmodel}
y_{g_k} = \hbf^H_{g_k} \Vbf_{g} \wbf_{g_k}s_{g_k} + \sum_{k' \neq
k} \hbf^H_{g_k} \Vbf_{g} \wbf_{g_{k'}} s_{g_{k'}} + \sum_{g' \neq
g} \sum_{j = 1}^{K_{g'}} \hbf_{g_k}^H  \Vbf_{g'} \wbf_{g'_j}
s_{g'_j} + n_{g_k},
\end{equation}
where the second and third terms in the right-hand side (RHS) of
\eqref{signalmodel} are the intra-group and inter-group
interference, respectively.  With the assumed ZF inner beamforming
the intra-group interference is completely eliminated, and  the
signal-to-interference-plus-noise ratio (SINR) at  user $k$ in
group  $g$ is given by
\begin{equation}  \label{eq:SINRformula}
\text{SINR}_{g_k} = \frac{ |\hbf^H_{g_k} \Vbf_{g}
\wbf_{g_k}|^2}{ \sum_{g' \neq g} \sum_{j = 1}^{K_{g'}}
|\hbf_{g_k}^H  \Vbf_{g'} \wbf_{g'_j}|^2 + \sigma^2}.
\end{equation}
Using the SINR  as the optimization criterion, one could try to
design the outer beamformer $\{\Vbf_g, ~g=1,\cdots,G\}$ to
maximize a relevant measure such as the sum rate. However, this
criterion generally leads to a challenging nonconvex optimization
problem since each user's SINR is jointly dependent on $\{\Vbf_g,
g=1,\cdots,G\}$ in a nonconvex manner
\cite{Sadek&Tarighat&Sayed:07WCOM,Zakhour&Gesbert:09WSA}.
 To circumvent the difficulty, we here adopt the SLNR approach
 proposed in \cite{Sadek&Tarighat&Sayed:07WCOM} and
 \cite{Zakhour&Gesbert:09WSA}. The SLNR approach considers the ratio between the signal power to the desired receiver and the power of the total
 interference to undesired receivers caused by the desired transmitter plus noise, not
 the power of the total interference received at the desired
 receiver. The rationale for this approach is that it is
 reasonable for the  transmitter to maximize the signal power to the desired receiver
 for a given allowed level of interference to undesired
 receivers in multi-user interference channels. The SLNR method is shown to be Pareto-optimal in the rate region in
 certain MIMO interference channel scenarios  \cite{Zakhour&Gesbert:09WSA,
 Park&Lee&Sung&Yukawa:13SP}.
Under the assumption of ZF inner beamforming with equal power
allocation, the  SLNR for user $k$ in group $g$ is given by
\begin{equation}  \label{eq:SLNRformula}
\text{SLNR}_{g_k} = \frac{ | \hbf^H_{g_k} \Vbf_{g} \wbf_{g_k} |^2 }
                    {  \sum_{g' \neq g}
\sum_{j = 1}^{K_{g'}} |\hbf_{g'_j}^H  \Vbf_{g} \wbf_{g_k} |^2  +
\sigma^2 }.
\end{equation}
Note that the key difference between \eqref{eq:SINRformula} and
\eqref{eq:SLNRformula} is that the SINR of user $g_k$ (at the
receiver side) is a joint function of $\{\Vbf_g, g=1,\cdots,G\}$,
whereas the SLNR of user $g_k$ (at the transmitter side) is a
function of only $\Vbf_g$ not of $\{\Vbf_{g^\prime}, g^\prime \ne
g\}$. Note also that the SLNR of user $g_k$ is a function of the
channel $\{\hbf_{g_k}, \hbf_{g'_j}, g' \neq g,
j=1,\cdots,K_{g'}\}$ and the inner beamformer $\wbf_{g_k}$ in addition to $\Vbf_g$.

\section{Outer Beamformer Design Criterion and Optimization}
\label{sec:OuterBeamformerDesign}

Recall that the main advantage of the two-stage beamforming
results from the fact that the outer beamformer is designed
without knowing the CSI of the actual channel $\{ \hbf_{g_k},
 k=1,\cdots, K_g, g=1,\cdots, G\}$ \cite{Adhikary&Nam&Ahn&Caire:13IT,
Nam&Adhikary&Ahn&Caire:14JSTSP,Lee&Sung:14ITsub,Chen&Lau:14JSAC,
Liu&Lau:14SP}. Thus, the outer beamformer should be designed based
only on the channel covariance matrices $\{\Rbf_g, g=1,\cdots,G\}$ and
this leads to using the {\em average} SLNR as our design
criterion. Hence,  we formulate the outer beamformer design
problem as follows:
\begin{equation} \label{Problem1}
{\cal{P}}_1  : ~~\Vbf_g^* ~=~ \mathop{\arg \max}_{\Vbf_g^H \Vbf_g
= \Ibf } ~~ \sum_{k = 1}^{K_g} \Ebb [ \text{SLNR}_{g_k}]
~~\text{for each}~~ g = 1, 2, \cdots, G.
\end{equation}
Although Problem ${\cal{P}}_1$ is conceptually simple, solving the
optimization problem is not straightforward. The first difficulty
is the derivation of the average SLNR since the random quantities (i.e., the channel vectors) are both in the numerator and the
denominator as seen in \eqref{eq:SLNRformula} and a closed-form expression of the average SLNR is
not available. To circumvent this difficulty,  we first derive a
lower bound on the average SLNR and maximize this lower bound on
the average SLNR under the constraint $\Vbf_g^H\Vbf_g=\Ibf$.

 For a
given outer beamformer $\Vbf_g$ and a given inner beamformer
$\Wbf_g$, the  SLNR of user $g_k$ averaged over  channel realizations
is lower bounded as follows:
\begin{align}
\Ebb \left\{\text{SLNR}_{g_k}\right\} &= \Ebb  \left\{\frac{ | \hbf^H_{g_k} \Vbf_{g} \wbf_{g_k} |^2 }
                    {  \sum_{g' \neq g}
\sum_{j = 1}^{K_{g'}} |\hbf_{g'_j}^H  \Vbf_{g} \wbf_{g_k} |^2  + \sigma^2 } \right\}\\
&\stackrel{(a)}{\ge}  \Ebb  \left\{ \frac{ | \hbf^H_{g_k} \Vbf_{g} \wbf_{g_k} |^2 }
                    {  \sum_{g' \neq g}
\sum_{j = 1}^{K_{g'}} ||\hbf_{g'_j}^H  \Vbf_{g} ||^2 || \wbf_{g_k} ||^2  + \sigma^2 } \right\}
\\
&\stackrel{(b)}{=}  \Ebb  \left\{ \frac{ | \hbf^H_{g_k} \Vbf_{g}
\wbf_{g_k} |^2 }
                    {  \sum_{g' \neq g}
\sum_{j = 1}^{K_{g'}} ||\hbf_{g'_j}^H  \Vbf_{g} ||^2  + \sigma^2 } \right\}\\
&\stackrel{(c)}{=} \Ebb  \left\{  | \hbf^H_{g_k} \Vbf_{g}
\wbf_{g_k} |^2  \right\} \Ebb \left\{\frac{  1 }
                    { \sum_{g' \neq g}
\sum_{j = 1}^{K_{g'}} ||\hbf_{g'_j}^H  \Vbf_{g} ||^2  + \sigma^2 } \right\} \nonumber \\
&\stackrel{(d)}{\ge} \Ebb  \left\{  | \hbf^H_{g_k} \Vbf_{g}
\wbf_{g_k} |^2  \right\} \frac{  1 }
                    { \Ebb \left\{ \sum_{g' \neq g}
\sum_{j = 1}^{K_{g'}} ||\hbf_{g'_j}^H  \Vbf_{g} ||^2  + \sigma^2 \right\}}  \nonumber \\
&=  \frac{ \Ebb  \left\{ | \hbf^H_{g_k} \Vbf_{g} \wbf_{g_k} |^2\right ] }
                    { \sum_{g' \neq g}
\sum_{j = 1}^{K_{g'}} \Tr( \Vbf_{g}^H \Ebb \{ \hbf_{g'_j}\hbf_{g'_j}^H \} \Vbf_{g} )   + \sigma^2 } \\
&\stackrel{(e)}{=} \frac{ \Ebb  \left\{ | \hbf^H_{g_k} \Vbf_{g}
\wbf_{g_k} |^2\right\} }
                    {
 \Tr( \Vbf_{g}^H  \Rbf_{g,2} \Vbf_{g} )  }. \label{lowerbound1}
\end{align}
Here, (a) follows from the sub-multiplicativity of  norm $||\Abf
\Bbf|| \le ||\Abf||||\Bbf||$; (b) follows from the equal power
allocation $||\wbf_{g_k}||=1$; (c) follows from the independence
between the desired-group channels and  other group channels; (d)
results from Jensen's inequality since  the function
$f(x)=\frac{1}{x}$ is convex for $x \ge 0$; and (e) follows from
\begin{equation}  \label{eq:Rbfg2}
\Rbf_{g,2} \defeq \sum_{g' \neq g} K_{g'} \Rbf_{g'} +
\frac{\sigma^2}{M_g}\Ibf.
\end{equation}

The next difficulty in deriving a lower bound on the average SLNR
is that the ZF inner beamformer $\wbf_{g_k}$ is based on the CSI
of the effective channel $\Hbf_g \Vbf_g$.
However, at the time of designing $\Vbf_g$ this effective channel
is not determined. This difficulty is properly circumvented by
exploiting the property of ZF inner beamforming and a lower bound
on the average signal power appearing in the numerator of the RHS
of  \eqref{lowerbound1} is given in the following theorem.

\vspace{0.5em}
\begin{theorem} \label{theorem:lowerbound}
 When the outer beamformer $\Vbf_g$ is given and the inner-beamformer  $\Wbf_g
= [ \wbf_{g_1}, \wbf_{g_2},$ $\cdots, \wbf_{g_{K_g}}]$ is a ZF
beamformer with equal power allocation,  $\Ebb \left[ |\hbf^H_{g_k} \Vbf_{g}
\wbf_{g_k} |^2 \right ]$ is lower-bounded as
\begin{equation}
\Ebb \left\{ |\hbf^H_{g_k} \Vbf_{g} \wbf_{g_k} |^2 \right\} \ge
\Tr(\Vbf_g^H \Rbf_g \Vbf_g) - (K_g - 1)\lambda_g,
\end{equation}
where $\lambda_g$ is the largest eigenvalue of the channel covariance matrix $\Rbf_g$ of group $g$. \\
\end{theorem}
\begin{proof}
 See Appendix.
\end{proof}

\vspace{0.5em}

 Applying Theorem \ref{theorem:lowerbound} to
\eqref{lowerbound1}, we obtain a lower bound on the average SLNR,
given by
\begin{align}
\Ebb \left[ \text{SLNR}_{g_k} \right]
&\ge  \frac{ \Ebb \left[| \hbf^H_{g_k} \Vbf_{g} \wbf_{g_k} |^2\right ] }
                    {
 \Tr( \Vbf_{g}^H  \Rbf_{g,2} \Vbf_{g} )  }  \nonumber\\
&   \ge  \frac{ \Tr(\Vbf_g^H \Rbf_g \Vbf_g) - (K_g - 1)\lambda_g  }
                    { \Tr \left (\Vbf_{g}^H  \Rbf_{g,2} \Vbf_{g} \right )  }  \nonumber \\
&=  \frac{ \Tr(\Vbf_g^H  \Rbf_{g,1} \Vbf_g)   }
                    { \Tr \left (\Vbf_{g}^H  \Rbf_{g,2}  \Vbf_{g} \right )  }
\end{align}
where $\Rbf_{g,1}$ is defined as
\begin{equation}  \label{eq:Rbfg1}
\Rbf_{g,1}
\defeq \Rbf_g - \frac{(K_g - 1)}{M_g}\lambda_g \Ibf.
\end{equation}
Note that the derived lower bound on the average SLNR depends only
on the group index $g$ not on the user index  $k$. This makes
sense since each user in the same group has the same channel
statistics. Finally, the proposed outer beamformer design based
only on the channel statistic information
$\{\Rbf_g,g=1,\cdots,G\}$ is formulated as

\begin{equation} \label{Problem2}
{\cal{P}}_2  : ~~\Vbf_g^* ~=~ \mathop{\arg \max}_{\Vbf_g^H \Vbf_g
= \Ibf } ~~ \frac{ \Tr(\Vbf_g^H  \Rbf_{g,1} \Vbf_g)   }
                    { \Tr \left (\Vbf_{g}^H  \Rbf_{g,2}  \Vbf_{g} \right )  } ~~\text{for each}~~ g = 1, 2, \cdots,
                    G,
\end{equation}
where $\Rbf_{g,1}$ and $\Rbf_{g,2}$ are given by \eqref{eq:Rbfg1}
and \eqref{eq:Rbfg2}, respectively.  Note that in Problem
\eqref{Problem2}, the outer beamformer design is performed for
each group separately. This is an advantage of the proposed design
method; complicated joint optimization is not required.

Note that in Problem ${\cal{P}}_2 $, $\Rbf_{g,1}$ is Hermitian and
$\Rbf_{g,2}$ is positive definite due to the added identity matrix
in $\Rbf_{g,2}$ in \eqref{eq:Rbfg2}. Problem ${\cal{P}}_2 $
maximizes the quotient (or ratio) of two traces under the
constraint $\Vbf_g^H \Vbf_g = \Ibf$, and is known as a {\em trace
quotient problem (TQP) or trace ratio problem (TRP)}, which is
often encountered in linear discriminant analysis (LDA) for
feature extraction and dimension reduction
\cite{Shen&Diepold&Huper:10MTNS,Zhang&Yang&Liao:14OL}.  Several
research works have been performed to understand  the theoretical
properties of TQP and to develop numerical algorithms for TQP
\cite{Wang&Yan&Xu&Tang&huang:07CVPR, Shen&Diepold&Huper:10MTNS,
Zhang&Yang&Liao:14OL}.

If we relax the orthonormality constraint $\Vbf_g^H \Vbf_g = \Ibf$
for the outer beamformer,  then Problem ${\cal{P}}_2 $ reduces to
the following simpler optimization problem:
\begin{equation} \label{Problem3}
{\cal{P}}_3  : ~~\Vbf_g^\star ~=~ \mathop{\arg \max} ~~ \frac{
\Tr(\Vbf_g^H \Rbf_{g,1} \Vbf_g)   }
                    { \Tr \left (\Vbf_{g}^H  \Rbf_{g,2}  \Vbf_{g} \right )  } ~~\text{for each}~~ g = 1, 2, \cdots, G.
\end{equation}
It was shown in \cite{Sadek&Tarighat&Sayed:07WCOM} that the
optimal $\Vbf_g^\star$ of Problem ${\cal{P}}_3$ is the $M_g$
generalized eigenvectors corresponding to the $M_g$ largest
generalized eigenvalues of the matrix pencil $\Rbf_{g,1}
-\lambda\Rbf_{g,2}$. There exist many available fast algorithms to
obtain the generalized eigenvectors of the positive definite
matrix pencil $\Rbf_{g,1} -\lambda\Rbf_{g,2}$
\cite{Golub:book,Li:13SIAM}. Note that the obtained generalized
eigenvectors from Problem ${\cal{P}}_3$ do not satisfy the
orthogonality constraint $\Vbf_g^H \Vbf_g = \Ibf$ in general. The
optimal solution $\Vbf_g^\star$ of Problem ${\cal{P}}_3 $ can be
decomposed by thin singular value decomposition (SVD) as
\begin{equation} \label{suboptimal}
\Vbf_g^\star = \Phibf_g \Dbf_g \Psibf_g,
\end{equation}
where $\Phibf_g$ is an $M \times M_g$  matrix satisfying
$\Phibf_g^H\Phibf_g = \Ibf$; $\Dbf_g$ is an $M_g \times M_g$
diagonal matrix; and $\Psibf_g$ is an $M_g \times M_g$ unitary
matrix.  By setting
\begin{equation} \label{eq:quotientTrace}
\Vbf_g = \Phibf_g,
\end{equation}
 we can obtain an outer beamformer  satisfying
the constraint $\Vbf_g^H \Vbf_g = \Ibf$. However, this approach
does not necessarily yield an optimal solution to Problem
${\cal{P}}_2$.

\subsection{The outer beamformer design algorithm and its convergence}

It is  known that there is no closed-form solution to TQP
\cite{Wang&Yan&Xu&Tang&huang:07CVPR, Shen&Diepold&Huper:10MTNS,
Zhang&Yang&Liao:14OL}.  In this section, we tackle the original
TQP ${\cal{P}}_2$, and present an iterative algorithm by modifying
the algorithm in \cite{Wang&Yan&Xu&Tang&huang:07CVPR} developed
for the real matrix case to the complex matrix case, to directly
optimize the objective function subject to the orthonormality
constraint. Denote the objective function of Problem ${\cal{P}}_2$
at the $n$-th iteration by
\begin{equation} \label{object}
\rho_n = \frac{ \Tr({\Vbf_g^{(n-1)}}^H  \Rbf_{g,1}
{\Vbf_g^{(n-1)}}) }
                    { \Tr \left ({\Vbf_g^{(n-1)}}^H  \Rbf_{g,2}  {\Vbf_g^{(n-1)}} \right )
                    },
\end{equation}
where  $\Vbf_g^{(n-1)}$ is the outer beamformer at the $(n-1)$-th
iteration.  The outer beamformer $\Vbf_g^{(n)}$ is updated by
solving the following trace difference problem:
\begin{equation}
  \Vbf_g^{(n)} = \mathop{\arg \max}_{\Vbf_g^H \Vbf_g = \Ibf  } \Tr({\Vbf_g}^H ( \Rbf_{g,1} - \rho_n \Rbf_{g,2})
  \Vbf_g),
\end{equation}
where $\rho_n$ is given by \eqref{object}. The outer beamformer
$\Vbf_g^{(n)}$ of the trace difference problem in the $n$-th
iteration is updated by the $M_g$ eigenvectors corresponding to
the $M_g$ largest eigenvalues of the Hermitian matrix $\Rbf_{g,1}
- \rho_n \Rbf_{g,2}$ for given $\rho_n$. This procedure is
iterated until the iteration converges. The proposed algorithm for
Problem ${\cal{P}}_2$ is summarized in Algorithm \ref{alg:OBDTQP}.

\begin{algorithm}[t]
\caption{Outer beamformer design by TQP based on \cite{Wang&Yan&Xu&Tang&huang:07CVPR}} \label{alg:OBDTQP}

\begin{algorithmic}
\REQUIRE  The channel covariance matrices $\{\Rbf_g,
~g=1,\cdots,G\}$,  the noise variance $\sigma^2$, and the stopping
tolerance $\epsilon
> 0$.

\STATE Construct $\Rbf_{g,1}$ and $\Rbf_{g,2}$ by \eqref{eq:Rbfg1}
and \eqref{eq:Rbfg2}.

\STATE Set $\rho_{-1}=-\infty$ and initialize $\Vbf_g^{(0)}$ by
\eqref{eq:quotientTrace}.

\FOR{$n=1,2,\cdots$}

\STATE Compute the trace ratio $\rho_n$ from  $\Vbf_g^{(n-1)}$:
\begin{equation}
\rho_n = \frac{ \Tr({\Vbf_g^{(n-1)}}^H  \Rbf_{g,1}
{\Vbf_g^{(n-1)}}) }
                    { \Tr \left ({\Vbf_g^{(n-1)}}^H  \Rbf_{g,2}  {\Vbf_g^{(n-1)}} \right )  }
\end{equation}

\STATE Solve the trace difference problem
\begin{equation}  \label{eq:Algo1update}
  \Vbf_g^{(n)} = \mathop{\arg \max}_{\Vbf_g^H \Vbf_g = \Ibf  } ~\Tr({\Vbf_g}^H ( \Rbf_{g,1} - \rho_n \Rbf_{g,2})
  \Vbf_g).
\end{equation}
That is, update $\Vbf_g^{(n)} = [ \vbf_{g,1}^{(n)},
\vbf_{g,2}^{(n)}, \cdots, \vbf_{g,M_g}^{(n)}] $ by solving
\begin{equation}
( \Rbf_{g,1} - \rho_n \Rbf_{g,2}) \vbf_{g,k}^{(n)} = \nu_k^{(n)}
\vbf_{g,k}^{(n)},
\end{equation}
where $\vbf_{g,k}^{(n)}$ is the eigenvector corresponding to the
$k$-th largest eigenvalue  $\nu_k^{(n)}$ of $\Rbf_{g,1} - \rho_n
\Rbf_{g,2}$.

\IF{$|\rho_n - \rho_{n-1} | < \epsilon$}

\STATE Break the loop.

\ENDIF

\ENDFOR

\end{algorithmic}
\end{algorithm}

The monotonic increase of the objective function $\rho_n$ by  Algorithm
\ref{alg:OBDTQP} is guaranteed by the following theorem:

\vspace{0.5em}

\begin{theorem} \label{theorem:convergence}
Algorithm 1 monotonically increases the trace quotient $\rho_n$, i.e.,  $\rho_{n+1} \ge \rho_{n}$ for all $n$. \\
\end{theorem}
\begin{proof}
See Appendix.
\end{proof}

\vspace{0.5em}

Theorem \ref{theorem:convergence} implies  that the proposed outer
beamformer design algorithm converges at least to  a local
optimum. Furthermore, the proposed algorithm actually yields a
globally optimal solution to maximize the trace ratio of Problem
${\cal{P}}_2$ under the orthogonality constraint. This is because
it was shown that any local optimum is also a global optimum for
TQP \cite{Shen&Diepold&Huper:10MTNS, Zhang&Yang&Liao:14OL}
and thus any convergent algorithm yields a global maximizer of
TQP.

\subsection{Discussion}

Now consider the difference between the proposed outer beamformer
design method here and the algorithm in
\cite{Chen&Lau:14JSAC}. (In Section
\ref{sec:numericalresults}, it will be seen that both algorithms
perform better than the outer beamformer design method in
\cite{Adhikary&Nam&Ahn&Caire:13IT}.) The algorithm in
\cite{Chen&Lau:14JSAC} basically minimizes
\begin{equation}  \label{eq:LauAlgo}
\left(\sum_{g^\prime \ne g} \Rbf_{g^\prime} \right) - w \Rbf_g,
\end{equation}
where $w$ is an arbitrary chosen constant weighting factor.
Problem \eqref{eq:LauAlgo} is equivalent to maximizing
\begin{equation}  \label{eq:LauAlgo1}
\Rbf_g - \lambda \left(\sum_{g^\prime \ne g} \Rbf_{g^\prime}
\right),
\end{equation}
where $\lambda = 1/w$. Note that the first difference is that
$\Rbf_g$ and $\sum_{g^\prime \ne g} \Rbf_{g^\prime}$ are
respectively used in \cite{Chen&Lau:14JSAC} instead of
$\Rbf_{g,1}$ and $\Rbf_{g,2}$ used in the proposed algorithm.
Thus, in the proposed algorithm there is  slight change in the
signal power by including the impact of ZF inner beamforming and
the inclusion of the thermal noise in the leakage part. However,
the major difference of the proposed approach from the existing
method in \cite{Chen&Lau:14JSAC} is the formulation of the trace
ratio in Problem $\Pc_2$. It was shown by Zhang {\it et al.} that
the optimal solution of TQP $\Pc_2$ is the $M_g$ dominant
eigenvectors of
\begin{equation} \label{eq:discussionP2}
\Rbf_{g,1} - \rho(\Vbf_g) \Rbf_{g,2}
\end{equation}
 with $\rho(\Vbf_g)~(=
{\mathbb{E}}\{ \mbox{SLNR}_{g_k}\})$ achieving its maximum value
\cite{Zhang&Yang&Liao:14OL}. Thus, the proposed algorithm
optimizes the weighting factor $\rho(\Vbf_g)$  for different
channel statistics  by solving a {\em nonlinear} eigenvalue
problem\footnote{This procedure is clearly seen in  Algorithm
\ref{alg:OBDTQP}.} (note that the weighting factor itself is a
function of the design variable $\Vbf_g$), whereas the existing
method solves a linear eigenvalue problem by simply fixing the
weighting factor. The capability of optimizing the weighting
factor of the proposed outer beamformer design approach can yield
significant performance gain over the existing method, as seen in
Section \ref{sec:numericalresults}.

Now consider the difference between  Problem $\Pc_2$ and Problem
$\Pc_3$. The optimal solution to Problem $\Pc_3$ is given by the
set of $M_g$ dominant generalized eigenvectors of
$(\Rbf_{g,1},\Rbf_{g,2})$ satisfying
\cite{Sadek&Tarighat&Sayed:07WCOM}
\begin{equation}  \label{eq:discussSadek}
\Rbf_{g,1}\xibf_i = \lambda_i \Rbf_{g,2}\xibf_i, ~~i=1,\cdots,M_g,
\end{equation}
 where $\lambda_1 \ge \cdots \ge \lambda_{M_g}$. Since
$\Rbf_{g,1}$ is Hermitian and $\Rbf_{g,2}$ is positive-definite,
$\{\xibf_i\}$ are $\Rbf_{g,2}$-orthogonal \cite{Parlett:book},
i.e.,
\[
\xibf_j^H \Rbf_{g,2} \xibf_i = 0, ~~i \ne j.
\]
Thus, $\Vbf_g=[\xibf_1,\cdots,\xibf_{M_g}]$ does not satisfy the
orthonormality constraint $\Vbf_g^H \Vbf_g = \Ibf$ unless
$\Rbf_{g,2}= c\Ibf$ for some constant $c$. However, the matrix in
\eqref{eq:discussionP2} is Hermitian and thus, it is
diagonalizable by unitary similarity, i.e., it has orthonormal
eigenvectors \cite{Horn&Johnson:12Matrixbook}. Hence, it yields an
outer beamformer satisfying the orthonormality constraint. Note
that \eqref{eq:LauAlgo1} involves a linear Hermitian eigen-system
composed of a weighted matrix difference with a fixed weighting
factor and Problem $\Pc_3$ involves a
 linear generalized eigen-system \eqref{eq:discussSadek} with a matrix pencil,
whereas the proposed problem $\Pc_2$ involves a nonlinear
Hermitian eigen-system \eqref{eq:discussionP2} again composed of a
weighted matrix difference but with a weighting factor depending
on the design variable.

\section{Numerical results}
\label{sec:numericalresults}

In this section, we provide some numerical results to evaluate the
performance of the proposed outer beamformer design method in
Section \ref{sec:OuterBeamformerDesign}.  Throughout the
simulation, we considered a massive multiple-input single-output
(MISO) downlink system in which a base station is equipped with a
ULA of $M = 128$ antenna elements and each of $K$ users has a
single receive antenna. The $K$ users were grouped into four
groups $(G=4)$, and the base station supported five users
simultaneously for each group, i.e., $K_g = 5$ for each  $g = 1,
2, 3,$ and $4$. The channel covariance matrix for each group is
specified by \eqref{onering} with the center angle $\theta$ and
the AS parameter $\Delta$. The channel vector for each user was
independently generated according to the model \eqref{channel}.
The stopping tolerance of the proposed outer beamformer design
algorithm was set as $\epsilon = 10^{-4}$. The noise power is set
as $\sigma^2 = 1$. From here on, all dB power values are relative
to $\sigma^2 = 1$.

\begin{figure}[h]
    \centerline{
\scalefig{0.65} \epsfbox{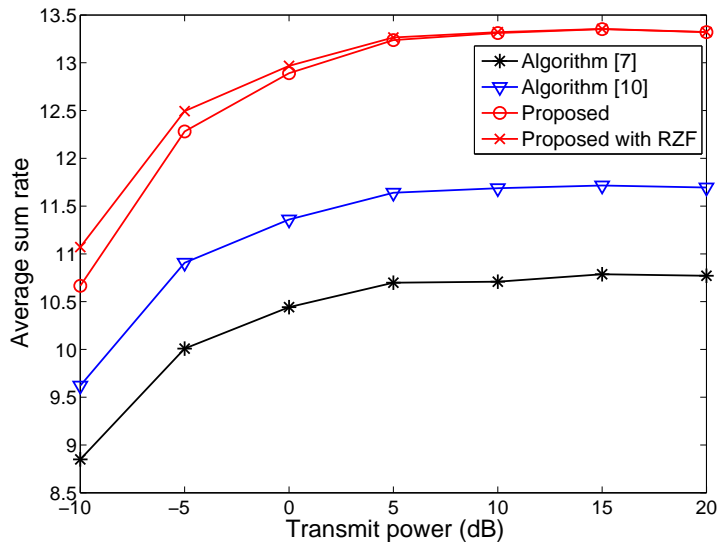} }
    \caption{Average sum rate versus transmit power}
    \label{fig:sumrate}
\end{figure}

Fig. \ref{fig:sumrate} shows the sum rate performance  of the
proposed outer beamformer design algorithm and two existing
algorithms in \cite{Adhikary&Nam&Ahn&Caire:13IT} and
\cite{Chen&Lau:14JSAC}. (Throughout the simulation, the weighting
factor for the algorithm in \cite{Chen&Lau:14JSAC} was set as $w =
1$.) We also considered  regularized ZF (RZF)
\cite{Peel&Hochwald&Swindlehurst:05TCOM} inner beamforming with
the same outer beamformer designed by the proposed algorithm under
the assumption of ZF inner beamforming. The regularization factor
was set as $\alpha = K/P_T$, which is approximately optimized
based on \cite{Peel&Hochwald&Swindlehurst:05TCOM,
Joham&Kusume&Gzara&Utschick:02SST}.  The four virtual sector
parameters were $\theta=-45^o, -15^o, 15,^o$ and $45^o$ with
$\Delta = \pi/13$ (i.e. $2\Delta=27.7^o$). The curves in the
figure is the average sum rate over  2000 independent channel
realizations according to the model \eqref{channel}.
 It is
seen that the proposed algorithm has significant gain over the
other two algorithms. It is also seen that in the low SNR region
the proposed outer beamformer combined with RZF inner beamforming
has performance gain over the proposed outer beamformer combined
with ZF inner beamforming as expected.

\begin{figure}[htbp]
\centerline{ \SetLabels
\L(0.25*-0.1) (a) \\
\L(0.75*-0.1) (b) \\
\endSetLabels
\leavevmode
\strut\AffixLabels{ \scalefig{0.5}\epsfbox{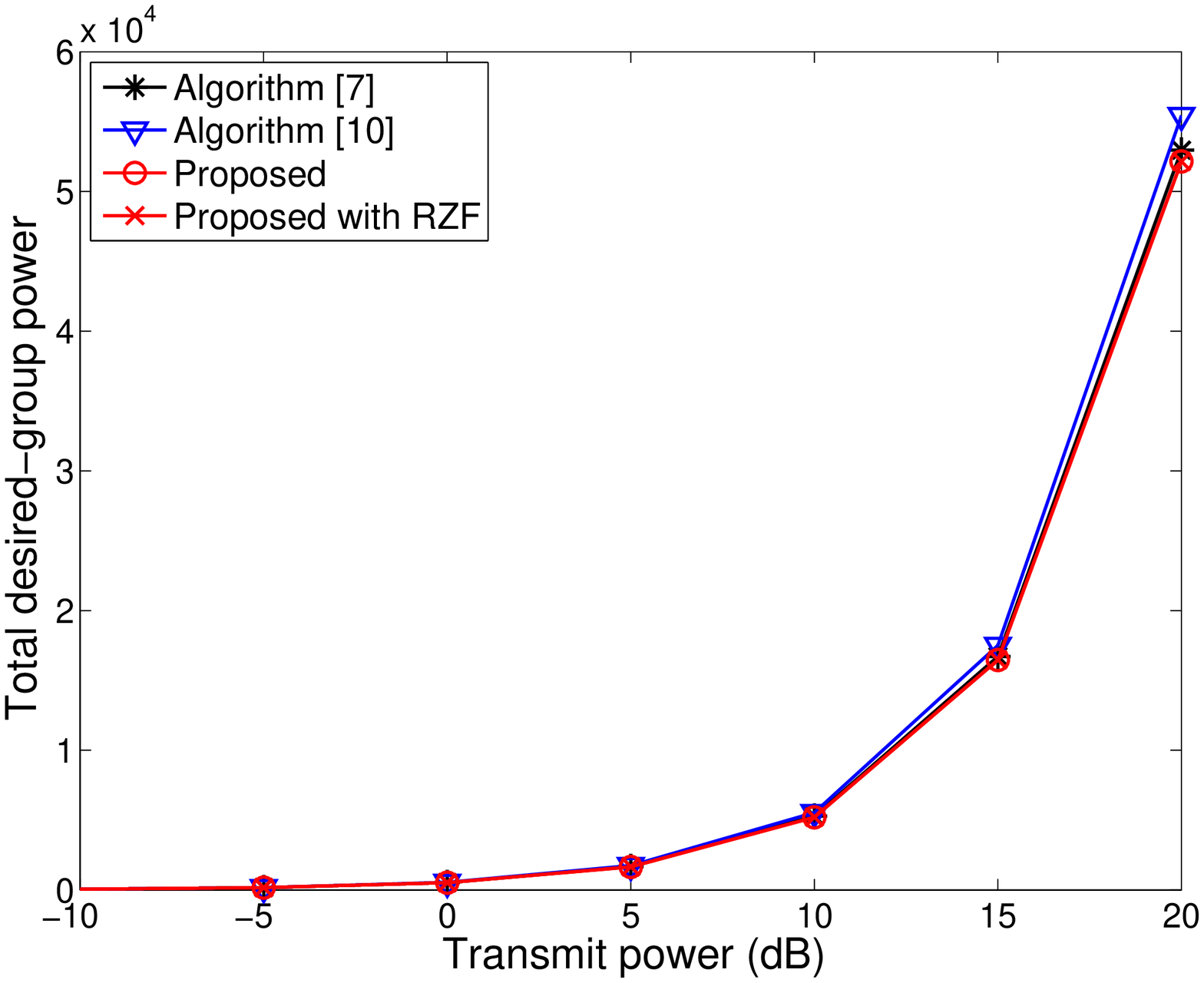}
\scalefig{0.5}\epsfbox{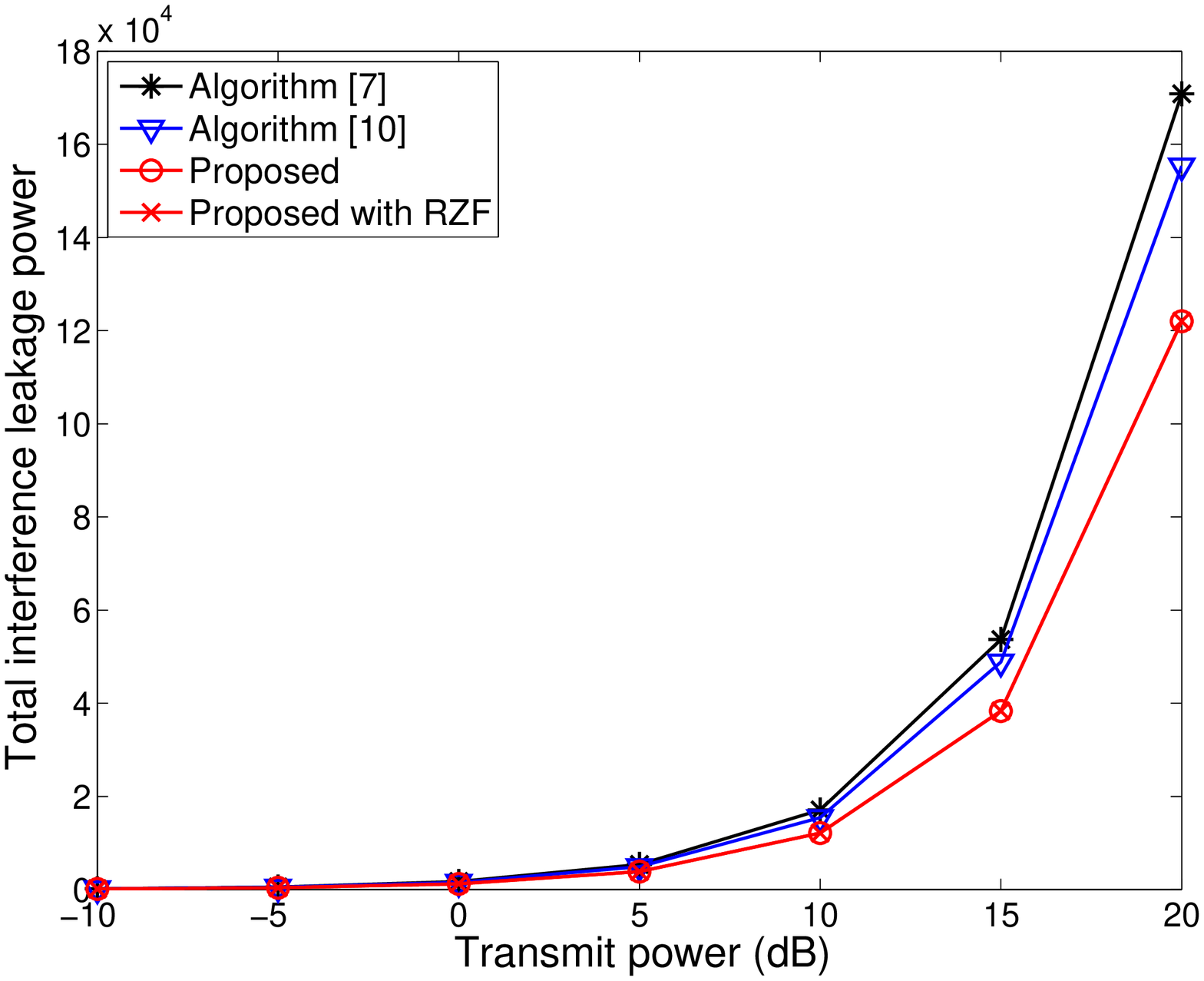} } }
\vspace{0.5cm} \caption{(a) total desired-group power versus
transmit power and (b) total interference leakage power versus
transmit power } \label{fig:sigalandleakage}
\end{figure}

Fig. \ref{fig:sigalandleakage} shows the total signal power  to
the desired group and the total leakage power to  undesired groups
under the same setting as in Fig. \ref{fig:sumrate}.  The shown
curves are average values over 2000 channel realizations. Now the
cause of the performance gain of the proposed algorithm over the
existing algorithms is clear. It achieves almost the same signal
power to the desired group with reduced leakage power to
undesired groups compared to the two other algorithms, by solving
a nonlinear eigenvalue problem with adaptive weighting resulting
from the trace quotient formulation for the signal power and the
leakage power.

\begin{figure}[tp]
    \centerline{
\scalefig{0.65} \epsfbox{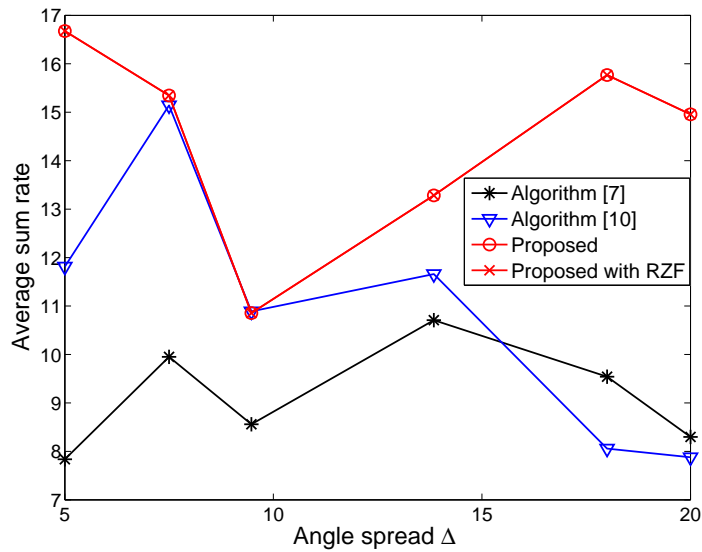} }
    \caption{Average sum rate over various algorithms versus angle spread (The transmit power is set as $P$ = 15dB.) }
    \label{fig:AS}
\end{figure}
 In the case of
the algorithm \cite{Chen&Lau:14JSAC}, the performance can change
with respect to  the weighting factor between the desired-group
signal power and the inter-group interference power. However, the
optimal weighting factor depending on the channel statistics is
not known {\em a priori} in the case of a constant weighting
factor, and the weight factor is not adaptively optimized
according to the given channel statistic information in
\cite{Chen&Lau:14JSAC}. Thus, we evaluated the sum rate
performance for different channel covariance matrix setup to see
the impact of a fixed weighting factor. The channel covariance
matrix is a function of the center angle and angular spread
according to the
 model \eqref{onering}.  We fixed the center angles of four virtual subsectors as  $\theta = -45^o, -15^o, 15,^o$ and $45^o$ but varied AS $\Delta$.
Fig. \ref{fig:AS} shows  the average sum rate performance with
respect to $\Delta$ when the transmit power is 15 dB.  It is seen
that at a certain AS, the fixed weighting factor $w=1$ yields good
performance but deteriorated performance for different AS's. Thus,
the optimization of the weighting factor depending on different
channel statistics is necessary for good performance for arbitrary
channel statistics.

\begin{figure}[tp]
    \centerline{
\scalefig{0.65} \epsfbox{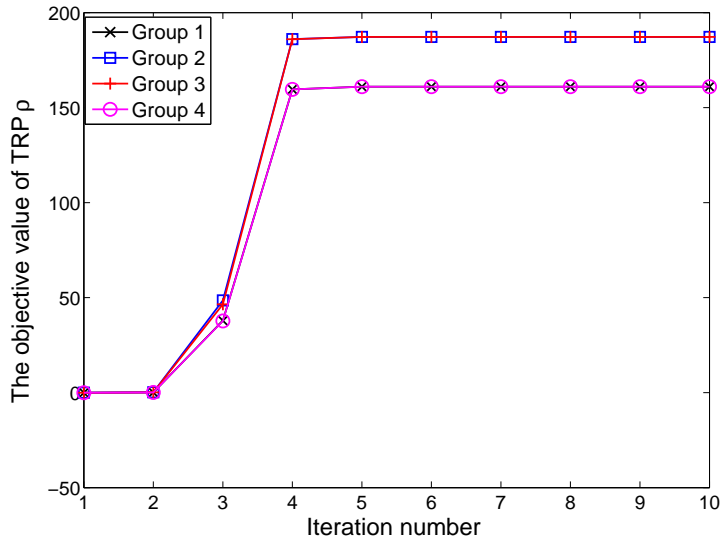} }
    \caption{The objective value of the TRP for each group versus iteration number }
    \label{fig:convergence}
\end{figure}

Finally, we tested the convergence of the proposed algorithm. Fig.
\ref{fig:convergence} shows that the objective value $\rho_n$ of
the TQP for each group by the proposed algorithm as a function of
iteration time. It is seen that the proposed algorithm
monotonically converges, as expected from Theorem
\ref{theorem:convergence}, and furthermore the proposed algorithm
converges after only a few iterations.

\section{Conclusion}
\label{sec:conclusion}

In this paper, we have considered the outer beamformer design
problem based  only on channel statistic information for two-stage
beamforming for massive MIMO downlink and have proposed an outer
beamformer design method which maximizes a lower bound on the
average SLNR. We have shown that the proposed SLNR-based outer
beamformer design problem reduces to a TQP, which is often
encountered in the field of machine learning, and have presented
an iterative algorithm to obtain an optimal solution to the
proposed TQP. Numerical results show that the proposed outer
beamformer design method yields significant performance gain over
existing methods. The proposed method can easily be adapted to the
multi-cell scenario in which each base station serves one group of
users.

\appendices
\section{Proof of Theorem \ref{theorem:lowerbound}}

Using similar techniques to those used in the proof of Lemma 1 in
\cite{Li&Paul&Narasimhan&Cioffi:06IT} and the proof of Theorem 1
in \cite{Liu&Lau:14SP}, we first prove several lemmas necessary
for proof of  Theorem \ref{theorem:lowerbound}.

\begin{lemma}  \label{lem:lemma1}
For ZF inner beamforming, we have
\begin{equation}
|{\hbf^H_{g_k} \Vbf_{g}}
\wbf_{g_k} |^2 = \hbf^H_{g_k} \Vbf_{g} \wt{\Ubf}_{g_k}
\wt{\Ubf}^H_{g_k} \Vbf_g^H \hbf_{g_k},
\end{equation}
 where  the column space of
$\wt{\Ubf}_{g_k}$ is the null space of the composite effective
channel matrix $\wt{\Hbf}_{g, -k}$ except the channel of user
$g_k$, i.e., $\wt{\Hbf}_{g, -k}  \defeq [ \Vbf_g^H \hbf_{g_1},
\cdots , \Vbf_g^H \hbf_{g_{k-1}},$ $ \Vbf_g^H \hbf_{g_{k+1}}
\cdots \Vbf_g^H \hbf_{g_{K_g}}]^H $.
\end{lemma}

\begin{proof} $\wt{\Hbf}_{g, -k} \in {\cal{C}}^{ (K_g -1) \times
M_g} ~(K_g \le M_g)$ is decomposed by  singular value
decomposition (SVD) as
\begin{equation} \label{eq:lem1hbfgmk}
\wt{\Hbf}_{g, -k} = \Upsilonbf_{g_k} \Dbf_{g_k} \left[
\begin{array} {c}
\Ubf^H_{g_k} \\
\wt{\Ubf}^H_{g_k}
\end{array}
\right ]
\end{equation}
where $\Upsilonbf_{g_k}$ is a $(K_g -1) \times (K_g -1)$ unitary
matrix, $\Dbf_{g_k}$ is a $M_g \times M_g$ diagonal matrix,
$\Ubf_{g_k}$ is a $M_g \times (K_g -1)$ submatrix whose column
space is the row space of $\wt{\Hbf}_{g, -k}$, and
$\wt{\Ubf}_{g_k}$ is a $M_g \times (M_g - K_g + 1)$ submatrix
whose column space is the null space of $\wt{\Hbf}_{g, -k}$. The
ZF inner beamformer $\wbf_{g_k}$  with equal power can be
expressed by the projection of the effective channel $\Vbf_g^H
\hbf_{g_k}$ of user $g_k$ to $\wt{\Ubf}_{g_k}$ as
\cite{Li&Paul&Narasimhan&Cioffi:06IT}
\begin{equation}  \label{eq:innerformula}
\wbf_{g_k}  = \frac{\wt{\Ubf}_{g_k} \wt{\Ubf}^H_{g_k} \Vbf_g^H
\hbf_{g_k}}{ ||\wt{\Ubf}_{g_k} \wt{\Ubf}^H_{g_k} \Vbf_g^H
\hbf_{g_k}||}.
\end{equation}
By  \eqref{eq:innerformula}, $|{\hbf^H_{g_k} \Vbf_{g}} \wbf_{g_k}
|^2$ is expressed
 as
\begin{eqnarray*}
|{\hbf^H_{g_k} \Vbf_{g}} \wbf_{g_k} |^2 &=& \hbf^H_{g_k} \Vbf_{g} \wbf_{g_k} \wbf_{g_k}^H \Vbf_{g}^H \hbf_{g_k} \\
&=&  \frac{\hbf^H_{g_k} \Vbf_{g} \wt{\Ubf}_{g_k} \wt{\Ubf}^H_{g_k}
\Vbf_g^H \hbf_{g_k} \hbf_{g_k}^H \Vbf_g \wt{\Ubf}_{g_k}
\wt{\Ubf}^H_{g_k} \Vbf_{g}^H \hbf_{g_k}}
          { ||\wt{\Ubf}_{g_k} \wt{\Ubf}^H_{g_k} \Vbf_g^H \hbf_{g_k}||^2 } \\
&=& \hbf^H_{g_k} \Vbf_{g} \wt{\Ubf}_{g_k} \wt{\Ubf}^H_{g_k}
\Vbf_g^H \hbf_{g_k}.
\end{eqnarray*}
This concludes the proof.
\end{proof}

\vspace{0.5cm}
\begin{lemma}  \label{lem:lemma2}
  The conditional expectation of $||{\hbf^H_{g_k} \Vbf_{g}} \wbf_{g_k} ||^2$ for given $\wt{\Hbf}_{g, -k}$ is given by
\begin{equation}
\Ebb \left\{ |{\hbf^H_{g_k} \Vbf_{g}} \wbf_{g_k} |^2 ~|~
\wt{\Hbf}_{g, -k} \right\} = \Tr( \Sigmabf_{g_k} )
\end{equation}
where
\begin{equation}  \label{eq:sigmabfgk}
\Sigmabf_{g_k} = \wt{\Ubf}_{g_k}^H \Vbf_g^H \Rbf_g \Vbf_g
\wt{\Ubf}_{g_k}
\end{equation}
 and $\wt{\Ubf}_{g_k}$ is defined in \eqref{eq:lem1hbfgmk} in Lemma
\ref{lem:lemma1}.
\end{lemma}

\vspace{0.5em} \begin{proof} By Lemma 1, we have $|{\hbf^H_{g_k}
\Vbf_{g}} \wbf_{g_k} |^2 = \hbf^H_{g_k} \Vbf_{g} \wt{\Ubf}_{g_k}
\wt{\Ubf}^H_{g_k} \Vbf_g^H \hbf_{g_k}$. For given $\wt{\Hbf}_{g,
-k}$, $\wt{\Ubf}_{g_k}$ is also given because $\wt{\Ubf}_{g_k}$
deterministically depends on  $\wt{\Hbf}_{g, -k}$ according to \eqref{eq:lem1hbfgmk}. Define
$\xibf_{g_k} \triangleq \wt{\Ubf}^H_{g_k} \Vbf_g^H \hbf_{g_k}$.
 Then,  $\xibf_{g_k}|\wt{\Hbf}_{g, -k}$ is
 a zero-mean
complex Gaussian random vector with the covariance matrix
$\Sigmabf_{g_k} = \wt{\Ubf}^H_{g_k} \Vbf_g^H \Rbf_g \Vbf_g
\wt{\Ubf}_{g_k}$, i.e., $\xibf_{g_k}|\wt{\Hbf}_{g, -k} \sim
{\cal{CN}}(0, \Sigmabf_{g_k})$, since
$\hbf_{g_k} \sim {\cal{CN}}(0, \Rbf_g)$ from \eqref{channel}. Let the eigenvalue decomposition of
$\Sigmabf_{g_k}$ be
\[
\Sigmabf_{g_k} = \ovl{\Ubf}_{g_k}
\ovl{\Lambdabf}_{g_k} \ovl{\Ubf}_{g_k}^H,
\]
where  $\ovl{\Lambdabf}_{g_k}=\mbox{diag}(\bar{\lambda}_{g_k,1},\cdots,\bar{\lambda}_{g_k,N_g})$ and $N_g = M_g - K_g + 1$.
 Then, $\xibf_{g_k}$ can be expressed
as
\[
\xibf_{g_k}|\wt{\Hbf}_{g, -k} =  \Sigmabf_{g_k}^{1/2} \etabf = \ovl{\Ubf}_{g_k}
\ovl{\Lambdabf}_{g_k}^{1/2}\etabf,
\]
where $\etabf \sim
{\cal{CN}}(0, \Ibf_{M_g -K_g +1})$. Thus, $|{\hbf^H_{g_k}
\Vbf_{g}} \wbf_{g_k} |^2$ can be rewritten as
\begin{align}
|{\hbf^H_{g_k} \Vbf_{g}} \wbf_{g_k} |^2 &= \hbf^H_{g_k} \Vbf_{g} \wt{\Ubf}_{g_k}
\wt{\Ubf}^H_{g_k} \Vbf_g^H \hbf_{g_k} \nonumber\\
&= \xibf_{g_k}^H \xibf_{g_k} = \mbox{Tr}(\xibf_{g_k} \xibf_{g_k}^H )\nonumber\\
&=   \mbox{Tr}( \ovl{\Ubf}_{g_k}
\ovl{\Lambdabf}_{g_k}^{1/2}\etabf
\etabf^H \ovl{\Lambdabf}_{g_k}^{H/2} \ovl{\Ubf}_{g_k}^H). \label{signalpower2}
\end{align}
  Using \eqref{signalpower2}, we now take
the conditional expectation of $|{\hbf^H_{g_k} \Vbf_{g}} \wbf_{g_k} |^2$  for given $\wt{\Hbf}_{g, -k}$:
\begin{align}
\Ebb \left\{ |{\hbf^H_{g_k} \Vbf_{g}} \wbf_{g_k} |^2 ~|~ \wt{\Hbf}_{g, -k} \right\}
&= \Ebb \left\{ \mbox{Tr}( \ovl{\Ubf}_{g_k}
\ovl{\Lambdabf}_{g_k}^{1/2}\etabf
\etabf^H \ovl{\Lambdabf}_{g_k}^{H/2} \ovl{\Ubf}_{g_k}^H) ~|~ \wt{\Hbf}_{g, -k} \right\} \nonumber \\
&=   \mbox{Tr} \left( \ovl{\Ubf}_{g_k}
\ovl{\Lambdabf}_{g_k}^{1/2} \Ebb \left\{ \etabf
\etabf^H  ~|~ \wt{\Hbf}_{g, -k} \right\} \ovl{\Lambdabf}_{g_k}^{H/2} \ovl{\Ubf}_{g_k}^H \right) \nonumber \\
&=   \mbox{Tr} \left( \ovl{\Ubf}_{g_k}
\ovl{\Lambdabf}_{g_k}^{1/2} \ovl{\Lambdabf}_{g_k}^{H/2} \ovl{\Ubf}_{g_k}^H \right) \nonumber \\
&=\sum_{i = 1}^{N_g} \ovl{\lambda}_{g_k,i}      = \Tr(\Sigmabf_{g_k}).
\end{align}
\end{proof}

 \vspace{0.5cm}

\begin{lemma}\cite{Horn&Johnson:12Matrixbook}[Corollary 4.3.18]  \label{lem:lemma3}
 Let $\Mbf$ be any $J \times J$ Hermitian matrix and let $j$ be a given integer such that $ 1\le j \le J$. Then, we have
\begin{equation}
\lambda_1  + \lambda_2 + \cdots + \lambda_{j} = \min_{\Ubf^H \Ubf = \Ibf_{j}} \Tr(\Ubf^H \Mbf \Ubf),
\end{equation}
and
\begin{equation}
\lambda_{J-j+1} + \lambda_{J-j+2} + \cdots + \lambda_{J} =
\max_{\Ubf^H \Ubf = \Ibf_{j}} \Tr(\Ubf^H \Mbf \Ubf),
\end{equation}
 where $\lambda_k$ is the $k$-th smallest eigenvalue of $\Mbf$
and equality holds  if the columns of the $J \times j$ matrix
$\Ubf$ are chosen to be the orthonormal eigenvectors associated
with the corresponding eigenvalues of $\Mbf$.
\end{lemma}

\vspace{0.5cm}
\begin{lemma}
For the matrix $\Sigmabf_{g_k}$ defined in Lemma \ref{lem:lemma2},  $\Tr(\Sigmabf_{g_k})$ is lower bounded as
\begin{equation}
\Tr (\Sigmabf_{g_k})  \ge  \Tr(\Vbf_g^H \Rbf_g \Vbf_g) - (K_g - 1)\lambda_g,
\end{equation}
where $\Vbf_g$ is the outer beamformer for group $g$ satisfying $\Vbf_g^H\Vbf_g=\Ibf$, $\Rbf_g$ is the channel covariance matrix of group $g$, and  $\lambda_g$ is the largest eigenvalue of $\Rbf_g$.
\end{lemma}

\vspace{0.5em}
\begin{proof}
  Applying the Rayleigh quotient technique, one can easily show that the maximum eigenvalue of the matrix $\Vbf_g^H \Rbf_g \Vbf_g$ is less than or equal to $\lambda_g$, where $\lambda_g$ is the largest eigenvalue of $\Rbf_g$, since $\Vbf_g$ has orthonormal columns. That is,
  \[
  \lambda_{max}(\Vbf_g^H \Rbf_g \Vbf_g) \le \max \frac{\xbf^H\Vbf_g^H \Rbf_g \Vbf_g\xbf^H}{\xbf^H\xbf} = \max \frac{\xbf^H\Vbf_g^H \Rbf_g \Vbf_g\xbf^H}{\xbf^H\Vbf_g^H\Vbf_g\xbf} = \max \frac{\tilde{\xbf}^H \Rbf_g \tilde{\xbf}^H}{\tilde{\xbf}^H\tilde{\xbf}} = \lambda_g.
  \]
       Define ${\mathbf{\Pi}}_g \triangleq \frac{1}{\lambda_g} \Vbf_g^H \Rbf_g \Vbf_g$ and let $ 0 \le \chi_1 \le \chi_2 \le \cdots,\chi_{M_g}$ be the eigenvalues of ${\mathbf{\Pi}}_g$. Then, the maximum eigenvalue $\chi_{M_g}$ of  ${\mathbf{\Pi}}_g$ is less than or equal to one from the above.
     Note from \eqref{eq:sigmabfgk} that $\frac{1}{\lambda_g} \Sigmabf_{g_k} = \wt{\Ubf}_{g_k}^H {\mathbf{\Pi}}_g
\wt{\Ubf}_{g_k}$, where $\wt{\Ubf}_{g_k}^H \wt{\Ubf}_{g_k}=\Ibf$.
  Applying  Lemma \ref{lem:lemma3} to $\frac{1}{\lambda_g} \Sigmabf_{g_k}$ with $\Mbf = {\mathbf{\Pi}}_g$, we have
\begin{align}
\frac{1}{\lambda_g} \Tr (\Sigmabf_{g_k}) &\ge \sum_{k = 1}^{N_g} \chi_k  \\
&\stackrel{(a)}{\ge} \sum_{k = 1}^{M_g} \chi_k - (K_g -1)\chi_{M_g} \\
&\stackrel{(b)}{\ge} \sum_{k = 1}^{M_g} \chi_k - (K_g - 1)  \\
&= \Tr( {\mathbf{\Pi}}_g ) - (K_g - 1)  \label{bound}
\end{align}
where  (a) follows from the relationship $N_g = M_g - K_g + 1$ and the fact that $\chi_{M_g}$ is the maximum eigenvalue of ${\mathbf{\Pi}}_g$ and  (b) follows from the fact that $\chi_{M_g} \le 1$. Multiplying both sides of \eqref{bound} by $\lambda_g$, we obtain the desired result.
\end{proof}

\vspace{0.5em}

\noindent {\it Proof of Theorem \ref{theorem:lowerbound}:}

Finally, we  prove Theorem \ref{theorem:lowerbound} by using the above lemmas. By Lemma 2 and Lemma 4, the conditional expectation of $|{\hbf^H_{g_k} \Vbf_{g}} \wbf_{g_k} |^2$ for given $\wt{\Hbf}_{g, -k}$ is lower-bounded as
\begin{eqnarray}
\Ebb \left\{ |{\hbf^H_{g_k} \Vbf_{g}} \wbf_{g_k} |^2 ~|~ \wt{\Hbf}_{g, -k} \right\} &=& \Tr( \Sigmabf_{g_k} ) \label{eq:pth1lowerbound1}\\
&\ge& \Tr(\Vbf_g^H \Rbf_g \Vbf_g) - (K_g - 1)\lambda_g. \label{eq:pth1lowerbound2},
\end{eqnarray}
where \eqref{eq:pth1lowerbound1} is by Lemma 2 and
\eqref{eq:pth1lowerbound2} is by Lemma 4. Note that the lower
bound \eqref{eq:pth1lowerbound2} is independent of $\wt{\Hbf}_{g,
-k}$ and is a constant. By taking expectation over $\wt{\Hbf}_{g,
-k}$ on both sides of \eqref{eq:pth1lowerbound2}. The left-hand
side (LHS) becomes $\Ebb \left\{ |{\hbf^H_{g_k} \Vbf_{g}}
\wbf_{g_k} |^2 \right\}$ by the law of iterated expectation and
the RHS does not change since the RHS is a constant. Hence, we
have
\begin{eqnarray}
\Ebb \left\{ |{\hbf^H_{g_k} \Vbf_{g}} \wbf_{g_k} |^2 \right\}
&\ge& \Tr(\Vbf_g^H \Rbf_g \Vbf_g) - (K_g - 1)\lambda_g
\end{eqnarray}
\hfill{$\blacksquare$}


\section{Proof of Theorem \ref{theorem:convergence}}

\tcr{}

As in the proof of Lemma 1 in
\cite{Wang&Yan&Xu&Tang&huang:07CVPR}, we first define  $f_n(
\Vbf_g) \defeq \Tr({\Vbf_g}^H ( \Rbf_{g,1} - \rho_n \Rbf_{g,2})
\Vbf_g)$. Then, we have $f_n( \Vbf_g^{(n-1)}) = 0$ from
\eqref{object}.
 The step \eqref{eq:Algo1update} of Algorithm 1 maximizes
$f_n(\Vbf_g)$ over the set $\{\Vbf_g: \Vbf_g^H \Vbf_g=\Ibf\}$
which includes $\Vbf_g^{(n-1)}$, and $\Vbf_g^{(n)}$ is the
maximizer of $f_n(\Vbf_g)$. Hence, we have
\begin{equation}
 f_n( \Vbf_g^{(n)}) \ge  f_n( \Vbf_g^{(n-1)}) = 0.
\end{equation}
Based on  the positive-definiteness of $\Rbf_{g,2}$, $f_n(
\Vbf_g^{(n)}) = \Tr({\Vbf_g^{(n)}}^H ( \Rbf_{g,1} - \rho_n
\Rbf_{g,2}) \Vbf_g^{(n)}) \ge 0$ can be rewritten as
\begin{equation}
\rho_{n+1}=\frac{\Tr({\Vbf_g^{(n)}}^H \Rbf_{g,1}
\Vbf_g^{(n)})}{\Tr({\Vbf_g^{(n)}}^H \Rbf_{g,2} \Vbf_g^{(n)})} \ge
\rho_n
\end{equation}
This concludes the proof. \hfill{$\square$}


\end{document}